\documentclass[12pt,english]{article}
\usepackage{amsfonts}
\usepackage[T1]{fontenc}
\usepackage{geometry}
\usepackage{mathrsfs}
\usepackage{csquotes}
\usepackage{color}
\usepackage{babel}
\usepackage{amsmath}
\usepackage{amssymb}
\usepackage{graphicx}
\usepackage{setspace}
\usepackage{esint}
\usepackage[
style=authoryear,
sorting=nyt,
sortcites=true,
maxcitenames=9,
maxbibnames=9,
backend=biber,
natbib,
citestyle=authoryear
]{biblatex}
\usepackage{tikz}
\usetikzlibrary{arrows.meta, positioning}
\usepackage[unicode=true,pdfusetitle,
 bookmarks=true,bookmarksnumbered=true,bookmarksopen=false,
 breaklinks=false,pdfborder={0 0 0},backref=false,colorlinks=true]{hyperref}
 \usepackage[nameinlink]{cleveref}
\usepackage{aecompl}
\usepackage{geometry}
\usepackage{babel}
\usepackage{amsthm}
\usepackage{paralist}
\usepackage{rotating}
\usepackage{tocloft}
\usepackage{xcolor}
\usepackage{bigints}
\usepackage[nottoc]{tocbibind}
\usepackage{bookmark}
\usepackage{lscape}
\usepackage{cleveref}
\usepackage{tikz}
\usepackage{pgfplots}
\usepackage{subcaption}
\usepackage{caption}
\usepackage{gnuplottex}
\usepackage{lmodern}
\usepackage[inline]{asymptote}
\usepackage{amsmath,amssymb,mathrsfs}
\usepackage{amsthm}
\usepackage{cleveref}
\usepackage{caption}
\hypersetup{
 linkcolor=darkred, citecolor=darkblue}
\makeatletter
\colorlet{darkred}{red!70!black}
\colorlet{darkblue}{blue!70!black}
\numberwithin{equation}{section}
\numberwithin{figure}{section}
\theoremstyle{plain}
\newtheorem{theorem}{Theorem}
\newtheorem{corollary}{Corollary}

\newtheorem{lemma}{\protect\lemmaname}
\newtheorem*{lemma*}{Lemma}
\newtheorem{proposition}{Proposition}
\newtheorem{definition}{Definition}
\theoremstyle{definition}
\newtheorem{remark}{Remark}
\newtheorem{example}{Example}
\usepackage{dcolumn}
\usepackage{etoc}
\usepackage{titlesec}
\titlespacing*{\section}
{0pt}{2.0ex plus 0.5ex minus 0.2ex}{0.8ex}

\titlespacing*{\subsection}
{0pt}{1.5ex plus 0.5ex minus 0.2ex}{0.6ex}

\titlespacing*{\subsubsection}
{0pt}{1.2ex plus 0.3ex minus 0.2ex}{0.5ex}

\newdimen\dummy
\dummy=\oddsidemargin
\allowdisplaybreaks

  \providecommand{\lemmaname}{Lemma}

\makeatother

\pgfplotsset{compat=1.18}

\DeclareMathOperator{\br}{br}

\begin{document}

\newpage

\title{Asymmetric Rationality: \\ Protocol Design and Behavioral Anomalies\footnote{This paper previously circulated under the title ``Strategic Learning with Asymmetric Rationality'' (arXiv:2510.23951). We thank Ignacio Esponda, Ariel Rubinstein, Andrea Wilson, Navin Kartik, Yeon-Koo Che, and seminar and conference participants at Columbia, Michigan, NYU, Penn, SITE, and Zhejiang University for helpful comments. We also thank the co-editor and three anonymous referees for their constructive comments and suggestions.}\\[1ex]  }

\author{Qingmin Liu\footnote{ql2177@columbia.edu, Department of Economics, Columbia University.}  \quad \quad  Yuyang Miao\footnote{ym2921@columbia.edu, Department of Economics, Columbia University. }}

\date{Revision: October 17, 2026}

\pagenumbering{roman}
\setcounter{page}{1}

\maketitle

\begin{abstract}

We study how asymmetric rationality in strategic interactions can generate seemingly anomalous patterns of belief updating and behavior. A boundedly rational receiver designs a finite-state decision protocol to learn from a fully rational, privately informed, and biased sender who controls the flow of information. The receiver must allocate scarce memory between learning and incentive provision. We identify a simple class of protocols sufficient for optimal design and show how their structure generates information disengagement, indecision, and polarization. The receiver hesitates before deciding but becomes unresponsive to further evidence afterward. Identical receivers following the same protocol can reach opposing decisions despite receiving the same sequence of signals, with the probability of disagreement remaining bounded away from zero as their payoffs approach the optimum. These behaviors, which might otherwise appear erratic or psychologically driven, arise from the strategic use of limited memory.

\bigskip


\end{abstract}

\newpage

{\small
\setstretch{0.87}
\tableofcontents
}

\newpage

\pagenumbering{arabic}
\setcounter{page}{1}

\section{Introduction}

Much of game theory is built on the assumption of symmetric perfect rationality, yet many strategic interactions are better understood as games with asymmetric rationality, in which players differ systematically in their capacity to process information and reason. Such asymmetries need not reflect intrinsic differences in cognitive ability; instead, they may arise from disparities in expertise, computational resources, or bandwidth. Think of policymakers relying on technically sophisticated think tanks or interest groups, regulators evaluating disclosures prepared by specialized firms, or voters learning from streams of narratives produced by professional media organizations, not to mention interactions between humans and AI. How can we model and study such games without introducing ad hoc behavioral, epistemic, or preference distortions? How should decision protocols be designed for a capacity-constrained agency facing more
sophisticated counterparts? What systematic patterns of behavior and belief updating emerge as equilibrium responses to asymmetric rationality, even though they may appear anomalous under standard game theory or otherwise be attributed to behavioral or psychological biases? 

This paper studies these questions in a canonical bilateral setting of learning and persuasion, in which a boundedly rational decision-maker (receiver) repeatedly receives information from a rational but biased expert (sender). The decision-maker faces a dual challenge: extracting useful information under asymmetric information while guarding against strategic manipulation under asymmetric rationality.

The preliminaries of our model are familiar.
The sender privately observes a binary state of the world, while the uninformed receiver starts with a prior belief. Their incentives are misaligned: the receiver aims to match his decision to the true state, whereas the sender always prefers a particular action, regardless of the truth.

This canonical setup is enriched along two dimensions. First, we model a dynamic process of information solicitation and provision in which both parties are active: the receiver can request additional evidence, the sender may voluntarily continue providing information over time, and either party may terminate the process. Such dynamics arise naturally when pharmaceutical firms seek FDA approval, interest groups lobby policymakers, or media outlets attempt to influence voters. In the model, the sender is constrained by a fixed signal-generating process that depends on the privately known state. The sender decides whether to generate a signal, but cannot fabricate, conceal, or freely design its realization. This captures settings in which, for example, the FDA imposes experimental guidelines or journalists are constrained by verifiable evidence.

Second, and central to our model, is the formulation of the receiver's bounded rationality. Rather than assuming perfect Bayesian updating, we model the receiver as designing a protocol subject to the constraint of a finite state space for processing information and making decisions, but otherwise free to optimize how incoming information is encoded, how memory states evolve, when decisions are made, and which actions are taken.\footnote{There is a fundamental question of whether bounded rationality can be modeled as rational choice over decision rules subject to constraints, and in what sense maximization over rules entails no loss of generality. We refer the reader to \citet{lipman1991decide} for a discussion of this question; see also \citet{simon1976substantive}, \citet{rubinstein1998modeling}, and \citet{compte2019ignorance} for important conceptual and modeling issues.} Such a protocol is often referred to as a finite automaton, a machine with finitely many memory states and one of the canonical models of computation \citep{hellman1970learning, hopcroft1979introduction}. Finite automata have also found important applications in repeated games \citep{rubinstein1986finite, abreu1988structure} and economics \citep{wilson2014bounded}.\footnote{There are many economic applications of automaton models; see, for example, \citet{rubinstein1998modeling},
\citet{kocer2010endogenous}, 
\citet{salant2011procedural}, 
\citet{chatterjee2022bounded},
and \citet{safonov2024slow}. \citet{compte2015plausible} and \citet{monte2014learning} study automata in games. Finite automata differ from finite-memory formulations in games with learning, as in \citet{liu2014limited} and \citet{pei2024reputation}, where the encoding of information and the transitions across memory states are more restrictive. They also differ from perceptron models with finitely many classifiers; see \citet{cho1996neural} for discussion. See also \citet{anderson2025} for a model that combines continuous Bayesian updating with finite memory, and \citet{ekmekci2011sustainable}, where information is strategically pooled to sustain reputation incentives, resulting in a finite rating system.} Our notion of bounded rationality differs from approaches that use finite automata to capture computational complexity in classical computer science \citep{hopcroft1979introduction} or strategic complexity in repeated games, and instead aligns with \citet{hellman1970learning} and \citet{wilson2014bounded}, where finite-state constraints limit the capacity to process information in learning.

The automaton model of learning is straightforward to formulate and intuitive to interpret. However, it is not without limitations: its analysis can be technically demanding, particularly when compared with behavioral economics models. Yet, the distinctive strength of \citet{hellman1970learning} lies in its remarkable characterization of the optimal attainable value from learning with a fixed memory size and the construction of automata that approach this bound, in contrast to asymptotic results obtained from heuristic finite-state algorithms as the memory constraint vanishes. This optimal design perspective is particularly relevant for economics. The central insight of \citet{wilson2014bounded} lies in its behavioral predictions, derived as outcomes of constrained optimization rather than posited as assumptions. Both papers and their many follow-ups can be formulated as single-agent decision problems with non-strategic information. To bring bounded learning into richer game-theoretic and economic environments, one must allow the information provider and the decision-maker to interact strategically, making asymmetric rationality a natural focal point.

 One of our main aims is to understand which decision protocols—and, consequently, which patterns of behavior and belief evolution—emerge as strategic responses to asymmetric rationality, and to identify the mechanisms that generate them. The scope of our applications extends beyond individual decision-making. An institutional player, such as the FDA or the SEC, may plausibly be boundedly rational relative to a highly specialized counterpart, such as a pharmaceutical company or an investment bank. One need only compare the resources and expertise that a pharmaceutical company devotes to developing and promoting a drug with those that the FDA can bring to bear in evaluating it, or contrast a regulatory board responsible for a broad range of issues with an interest group focused on a single policy agenda.

Under the institutional interpretation, a designer designs an institution or decision protocol for a capacity-constrained agency facing more sophisticated counterparts. This is perhaps one of the fundamental motivations for designing decision protocols. Under the individual decision-maker interpretation, the protocol may instead be interpreted \textit{as if} it had been shaped by evolutionary or adaptive processes, while the burden of optimization lies with the researcher rather than the boundedly rational agent. A third interpretation treats memory or decision states as distinct organizational units, or as individuals within a team. The protocol can then be viewed as designing an organization for processing strategically provided information.\footnote{See \citet{marschak1972economic} and \citet{radner1993organization} for discussions of the information-processing role of teams and organizations.}

The central tradeoff in protocol design is between enabling the
receiver to learn and incentivizing the sender to  providing
information. Both tasks draw on the same finite state space and
therefore compete for scarce memory capacity. We identify a simple class of protocols, termed \textit{parsimonious}, that are optimal among all protocols. A parsimonious protocol features two absorbing memory states corresponding to distinct final decisions, while all other states are transient and prescribe actions unfavorable to the sender. Absorbing states entail an informational loss by halting further learning, but provide an incentive benefit: if the sender ceases to provide information before absorption, an unfavorable action follows with probability one, which is strictly worse for the sender than the lottery over the two absorbing states. Through this design, the receiver gains full control over the flow of information, compelling the sender to behave as if non-strategic. Despite the intricate transition rules of optimal or near-optimal parsimonious protocols, we show that the receiver can achieve the same payoff as in the non-strategic benchmark with one fewer memory state, precisely quantifying both the incentive cost of guarding against strategic persuasion and the scope for manipulation.

The architecture of optimal parsimonious protocols gives rise to distinctive behavioral patterns that might otherwise be interpreted as erratic or psychologically driven. At the individual level, the decision-maker exhibits hesitation near the decision point but information shutoff once an opinion is formed. At the group level, common information need not lead to common opinions or decisions, a strong form of polarization. Their combination is particularly striking. These patterns arise from a common underlying mechanism rather than reflecting unrelated behavioral phenomena. Perhaps more importantly, because these non-standard patterns of behavior and belief updating arise as optimization under memory constraints rather than as mistakes or biases, they need not disappear with training or repeated experience.

When an absorbing memory state is reached,  which occurs almost surely, the decision-maker fully commits to a position and ceases to attend to further evidence---a behavior often understood as strategic ignorance or information disengagement. This absorption arises as an optimal adaptation to strategic considerations: remaining perpetually open to new information would effectively cede decision-making power to the sender, resulting in full manipulability. Therefore, the structure of optimal protocol, along with the consequent behavioral implications and, crucially, their underlying mechanisms, stand in contrast to those established in \citet{wilson2014bounded}, \citet{hellman1970learning}, and subsequent work with non-strategic information, which feature recurrent memory states with varying likelihoods. In fact, their optimal or near-optimal automata are fully manipulable by a strategic sender, who can induce her preferred decision with probability one.\footnote{\label{footnote:hellman-wilson}There is a subtle but important distinction between \citet{wilson2014bounded} and \citet{hellman1970learning}. \citet{wilson2014bounded} assumes a positive exogenous stopping rate $\eta$ in each period, and shows that exact optimal automata exist and converge to one with absorbing states as $\eta\rightarrow 0$. However, this limiting automaton is not near-optimal for the underlying decision problem in the limit; it can also be shown that away from the limit, the optimal automata in \citet{wilson2014bounded} remain manipulable by a strategic sender who discounts at the same rate $\eta$, since the transition probability out of the extreme states, on the order of $O(\eta^{1/2})$, is still too large.} It is important to note that stopping in this mechanism is not driven by time preferences, as the model eliminates discounting and all other frictions aside from asymmetric information and rationality, thereby highlighting a new incentive for stopping.

According to this prediction of the model, voters who ``make up their minds'' and tune out further media coverage, and policymakers who disregard additional scientific data after reaching a judgment---even when information remains freely provided or costless to acquire---may both be acting optimally, favoring early commitment over continued susceptibility to influence. Of course, ex post, their decisions may turn out to be mistaken, and additional information could still be valuable. 

In general, an optimal parsimonious protocol features stochastic transitions to the absorbing states. Consequently, both absorbing states can be reached conditional on the \textit{same} sequence of realized signals. Thus, independent voters with the same level of sophistication may become convinced of opposing policy views while watching the same media, and neither would relinquish their view for fear of being manipulated by the strategic media.  We show that this strong form of opinion polarization must occur with a strictly positive probability.

The model also predicts interesting behavior near the decision point that resembles ``cold feet''. As the receiver's memory state approaches an absorbing state---indicating an imminent final opinion or decision---his behavior becomes conservative. The probability of transitioning to the absorbing state becomes very small: a large volume of confirmatory signals is needed to convince him to make a decision.

Behaviors near the decision point---such as last-minute doubt, decision-closure anxiety, negativity bias, and elevated risk aversion near commitment---have been extensively discussed in behavioral science and psychology. A variety of explanations for such forms of ``indecision'' have been proposed and experimentally tested; however, to our knowledge, none distinguishes between strategic and non-strategic forms of information transmission. The model developed here provides a framework for interpreting these behavioral regularities as manifestations of equilibrium reasoning rather than as psychological biases. If one takes seriously the evolutionary shaping of the human brain, evolution itself can be viewed as the force driving this equilibration. A plausible hypothesis is that strategic information transmission has shaped the same neural architecture that is also used for decision-making with non-strategic senders.  More systematic experimental research is needed to investigate these behaviors further.

\section{Model}\label{sec:model}

\subsection{Basic Setup}

The sender (she) privately observes the state of nature $\theta \in \Theta=\{H,L\}$, with the prior belief $\Pr(\theta=H)=p\in (0,1)$. The receiver (he) faces a single decision problem in which he takes an action in $\{H,L\}$. His payoff is 1 if the action matches the state of nature and 0 otherwise. The sender's payoff is 1 if the action is $H$ and 0 if the action is $L$, regardless of the state of nature.

Before making his decision, the receiver can learn from the sender over time. In each period $t=0,1,...$, the sender decides whether to generate a public imperfect signal $s_t \in S$, but cannot fabricate the signal. The signal-generating process is i.i.d. conditional on the true state. We assume $S$ is finite and the signal distribution has full support: $\pi^\theta(s):=\textup{Pr}(s|\theta)>0$ for all $s\in S$ and $\theta\in \Theta$. Assume also $\pi^H\neq\pi^L$. 

The receiver is boundedly rational with a finite set of states \(M\),
which we refer to as memory states. We write
\(M=\{1,2,\ldots,m\}\), where $m\geq 2$ is exogenously given. His strategy is Markov on the state space \(M\)
and is represented by a finite automaton, which
we call a \textbf{protocol}, denoted by \(\Pi=(f,g,a)\), where:
\begin{itemize}
    \item \( f: M \times S \rightarrow \Delta(M) \) is the \textbf{transition function}: given the current memory state \( i \) and signal \( s \), \( f(i, s)(j) \) specifies the probability of transitioning to memory state \(j \).
    \item \( g \in \Delta(M) \) is the \textbf{initial distribution}, defining the probability of starting in each memory state at period 0.
    \item \( a: M \rightarrow [0, 1] \) is the \textbf{action rule}, determining the probability of taking action \( H \) in a memory state when the decision is called for.
\end{itemize}

\begin{example}\label{example:leading}
 A leading example is the symmetric binary case: $p=\frac{1}{2}$, $S=\{h,l\}$, and $\pi^H(h)=\pi^L(l)=q\in (\frac{1}{2},1).$ \Cref{fig:3-state} depicts a protocol on three memory states
$M=\{1,2,3\}$, with initial state $2$, so $g(2)=1$. Transitions are
deterministic: signal $h$ moves the state up and signal $l$ moves it
down, except that state $3$ is absorbing and hence unresponsive to
further signals. The action rule is given by
$a(1)=0$, $a(2)=\frac{1}{2}$, and $a(3)=1$.

Thus, if the sender stops immediately, the game ends in state $2$ and
the receiver randomizes equally between $H$ and $L$. If instead the
sender generates a signal, realization $h$ moves the state to the
absorbing state $3$, where the receiver ceases processing further
signals and action $H$ is taken with probability one; realization $l$
moves the state to state $1$, where the sender's stopping induces action
$L$, while continued information provision can move the state away from
state $1$ again.
\begin{figure}[ht]
\centering
 \begin{tikzpicture}[->, >={Stealth[length=1.5mm, width=1.5mm]}, semithick, node distance=2.7cm]

  \node[circle, draw, minimum size=1cm] (S1) {$0$};
  \node[circle, draw, minimum size=1cm, right of=S1] (S2) {$\frac{1}{2}$};
  \node[circle, draw, minimum size=1cm, right of=S2] (S3) {$1$};

  \node[below=1mm of S1] {1};
  \node[below=1mm of S2] {2};
  \node[below=1mm of S3] {3};

  \coordinate[above=0.6cm of S2] (init);
  \draw[->] (init) -- (S2);

  \path
    (S1) edge[bend left=20] node[above] {$h$} (S2)
     (S2) edge[bend left=0] node[above] {$h$} (S3)
    (S2) edge[bend left=20] node[below] {$l$} (S1)
    (S1) edge[loop above] node {$l$} (S1)
    (S3) edge[loop above] node {$h,l$} (S3);

\end{tikzpicture}
\caption{A three-state protocol: the numbers inside the circles indicate the probability that action $H$ is taken upon stopping in each memory state; arrows represent transitions triggered by signals; the initial memory state is 2.}
\label{fig:3-state}
\end{figure}
\end{example}

The timing of the game is as follows. At the beginning of the game,
the receiver chooses a protocol $\Pi=(f,g,a)$. The sender observes
the protocol and, during play, the receiver's current memory
state.\footnote{This captures, for example, a regulator or organization with a transparent decision-making process, or a setting in which a
sophisticated media outlet or lobbyist observes the audience's current
state of thinking. \Cref{sec:discussion} examines the robustness of our protocol designs to alternative information structures and history-dependent sender strategies.} 
We focus on stationary Markov strategies for the sender,
\[
\sigma:M\times\Theta\rightarrow[0,1],
\]
where $\sigma(i,\theta)$ specifies the probability of stopping the
signal process in memory state $i\in M$ when the state of nature is
$\theta\in\Theta$.

The initial memory state $m_0$ is drawn according to $g$. In each period $t$ when the protocol's current memory state is $m_t$
and a signal $s_t$ is generated, the memory state is updated from
$m_t$ to $m_{t+1}$ with probability $f(m_t,s_t)(m_{t+1})$. The game \textit{ends} in memory state $m_t$ in period $t$ if either the sender stops the signal process in that period, prompting the receiver to make a decision, 
or if the receiver's memory state ceases to update in response to any signal, i.e., $f(m_t, s_t)(m_t) = 1$ for all $s_t \in S$. When the game ends, the receiver implements action $H$ with probability $a(m_t)$ and action $L$ otherwise, and the players then receive their respective payoffs. If the game never ends, no decision is made and both players receive a payoff of zero. There is no discounting.

A protocol \( \Pi \) induces a Markov process on the set of memory states \( M \) for each $\theta \in \Theta$, under the exogenous i.i.d.\ signal process generated by $\pi^\theta$. For clarity and brevity, we adopt the terminology of Markov processes---such as absorbing memory state, transient memory state, and communicating class---as referring directly to the protocol \( \Pi \). 

\subsection{Expected Payoffs}\label{sec:expected payoff}

Given the receiver's protocol $\Pi$  and the sender's strategy $\sigma,$ denote by  $U^S(\Pi,\sigma)$ and $U^R(\Pi,\sigma)$ the sender's and receiver's expected payoffs, respectively. 
For each $\Pi$, let 
\begin{equation}
    \br(\Pi):=\arg\max\nolimits_{\sigma} U^S(\Pi,\sigma)
\end{equation}  
be the sender's best responses to $\Pi.$
The objectives of interest to the receiver are 
\begin{equation}\label{eq:objective1}
    \sup_{\Pi}\sup_{\sigma \in \br(\Pi) } U^R(\Pi,\sigma),
\end{equation}
\begin{equation}\label{eq:objective2}
    \sup_{\Pi}\inf_{\sigma \in \br(\Pi) } U^R(\Pi,\sigma).
\end{equation}
The two will be shown to coincide, but both must be carried through much of the analysis.

\section{Examples}\label{sec:eg}

  Consider again the symmetric binary case in \Cref{example:leading}.
The protocol in \Cref{fig:3-state} is fully manipulable by the sender
in the following sense: the sender can have her preferred action $H$
implemented regardless of $\theta$ by continuing to generate signals
until the receiver's protocol hits state 3, which happens in finite time
almost surely. To avoid such manipulation, the receiver can reduce the
probability with which action $H$ is played in state 3, but this does not affect
the sender's ability to reach state 3, which hurts the sender without
benefiting the receiver himself. The receiver can also use stochastic
transitions, but if sufficiently long streaks of signal $h$ can drive the
protocol to states favorable to action $H$, the sender can wait for such a
streak and manipulate the receiver.

There are many other modifications the receiver can make to the protocol. \Cref{fig:3-state-2-absorbing} describes a protocol that is immune to the aforementioned manipulation. 
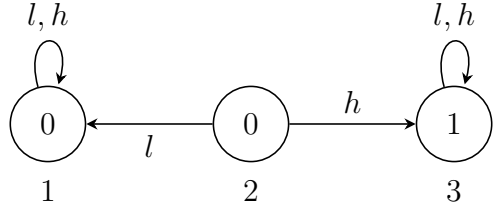
\begin{figure}[ht]
\centering
\begin{tikzpicture}[->, >={Stealth[length=1.5mm, width=1.5mm]}, semithick, node distance=2.7cm]

  \node[circle, draw, minimum size=1cm] (S1) {$0$};
  \node[circle, draw, minimum size=1cm, right of=S1] (S2) {$0$};
  \node[circle, draw, minimum size=1cm, right of=S2] (S3) {$1$};

  \node[below=1mm of S1] {1};
  \node[below=1mm of S2] {2};
  \node[below=1mm of S3] {3};

  \coordinate[above=0.6cm of S2] (init);
 \draw[->](init) -- (S2);

  \path
    (S2) edge[bend left=0] node[above] {$h$} (S3)
    (S2) edge[bend left=0] node[below] {$l$} (S1)
    (S1) edge[loop above] node {$l, h$} (S1)
    (S3) edge[loop above] node {$l,h$} (S3);
\end{tikzpicture}
\caption{A three-state protocol with two absorbing states; the initial state is transient.}
\label{fig:3-state-2-absorbing}
\end{figure}
Facing this protocol, the sender's best response is to continue sending
signals until an absorbing state is reached. However, this protocol
effectively uses only one signal realization. Starting from state 2,
under $\theta=H$ (or $L$), the probability of reaching absorbing state
3 is $q$ (or $1-q$). Thus, the receiver's expected payoff is $q$.
Is this the best the receiver can do with three memory states?
\Cref{cor: receiver-value-sym} of \Cref{thm:supUR} answers this
affirmatively.

The 3-state protocol described in \Cref{fig:3-state-2-absorbing} belongs to a particularly simple class that we call \textit{parsimonious} protocols. A parsimonious protocol has two absorbing states, all other states are transient, and the sender’s preferred action is prescribed only in one absorbing state, with the alternative action prescribed everywhere else. \Cref{thm:parsimonious} shows that every non-parsimonious protocol can be weakly improved upon, in terms of the receiver’s induced equilibrium payoff, by a parsimonious protocol.\footnote{\label{footnote:parsimony}In \Cref{fig:3-state-2-absorbing}, we can increase the action
probability at memory state~2 from $a(2)=0$ to a sufficiently
small positive value without changing equilibrium play.
Parsimony thus captures the underlying incentive structure
in a particularly stark form.}

However, the optimality of deterministic transitions is a special feature of the 3-state case, as formally confirmed by \Cref{thm:supUR}. \Cref{fig:5-state-random-chain} presents a parsimonious protocol with five memory states and random transitions. 

\begin{figure}[ht]
\centering
\begin{tikzpicture}[->, >={Stealth[length=1.5mm, width=1.5mm]}, semithick, node distance=2.1cm]

  \node[circle, draw, minimum size=1cm] (S1) {$0$};
  \node[circle, draw, minimum size=1cm, right of=S1] (S2) {$0$};
  \node[circle, draw, minimum size=1cm, right of=S2] (S3) {$0$};
  \node[circle, draw, minimum size=1cm, right of=S3] (S4) {$0$};
  \node[circle, draw, minimum size=1cm, right of=S4] (S5) {$1$};

  \node[below=1mm of S1] {1};
  \node[below=1mm of S2] {2};
  \node[below=1mm of S3] {3};
  \node[below=1mm of S4] {4};
  \node[below=1mm of S5] {5};

  \path
    (S1) edge[loop above] node {$h, l$} (S1)
    (S5) edge[loop above] node {$h, l$} (S5)

    (S2) edge[loop above] node {$(1\! \!- \! \epsilon)l$} (S2)
    (S4) edge[loop above] node {$(1\! \!- \! \epsilon)h$} (S4)

    (S2) edge[bend left=20] node[above] {$h$} (S3)
    (S3) edge[bend left=20] node[above] {$h$} (S4)
     (S4) edge[bend left=0] node[above] {$\epsilon h$} (S5)

    (S4) edge[bend left=20] node[below] {$l$} (S3)
    (S3) edge[bend left=20] node[below] {$l$} (S2)
    (S2) edge[bend left=0] node[above] {$\epsilon l$} (S1);

  \coordinate[above=0.6cm of S3] (init);
 \draw[->](init) -- (S3);

\end{tikzpicture}
\caption{A five-state protocol with two absorbing states. The transitions in states 2 (upon signal $l$) and 4 (upon signal $h$)  are random.}
\label{fig:5-state-random-chain}
\end{figure}
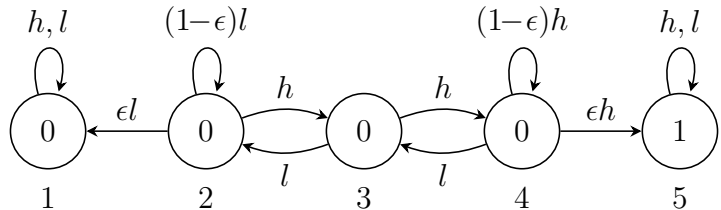

 For $\epsilon \in (0,1]$, starting from state 3, the ratio of absorption probabilities between the more likely and less likely absorbing states, under either state of nature, is
\[
\frac{q^{2}}{(1-q)^{2}}\cdot
\frac{q+(1-q)\epsilon}{1-q+q\epsilon}.
\]

The deterministic parsimonious protocol with
$\epsilon=1$ generalizes the optimal 3-state protocol in \Cref{fig:3-state-2-absorbing}. Under the deterministic protocol, the ratio becomes
$\frac{q^{2}}{(1-q)^{2}}$, and the receiver's expected payoff is \(
\frac{q^{2}}{q^{2}+(1-q)^{2}}.\) The deterministic protocol effectively makes use of two excess signals. 

As $\epsilon \rightarrow 0$, this ratio approaches its supremum
$\frac{q^{3}}{(1-q)^{3}}$, and the receiver's expected payoff approaches
$\frac{q^{3}}{q^{3}+(1-q)^{3}}>\frac{q^{2}}{q^{2}+(1-q)^{2}}$.
\Cref{cor: receiver-value-sym} of \Cref{thm:supUR} confirms that this is
the receiver's optimal payoff.
The stochastic protocol effectively utilizes three excess signals in the limit. Also, in the limit, the probability of reaching either absorbing
states is strictly between zero and one as
$\epsilon\rightarrow0$. Thus, information disengagement persists as
the receiver's payoff approaches its optimal value, with a non-vanishing
probability of committing to a mistaken decision.
\Cref{cor:interior-absorption} establishes this property more generally. By contrast, there is  a
discontinuity at $\epsilon=0$: the 5-state protocol reduces to a 3-state recurrent
protocol. 

In this class of protocols, transitions to the absorbing states occur only from their ``adjacent'' states, and the
corresponding transition probabilities vanish as $\epsilon\rightarrow0$. This reflects indecision
or hesitation near the decision points. \Cref{thm:behavioral_implication}
defines precisely the meaning of ``adjacent'' states and shows that these
properties extend beyond this example.

The stochastic transitions make it possible for two independent copies of the protocol, conditional on the same infinite sequence of signals, to reach different absorbing states, generating a strong form of polarization. The probability of such disagreement is bounded away from zero even if $\epsilon\rightarrow 0$. \Cref{thm:polarization} confirms this for any asymptotically optimal sequence of parsimonious protocols. 

Finally, it is useful to contrast this protocol with those characterized by \citet{hellman1970learning} and \citet{wilson2014bounded} in non-strategic learning environments. \Cref{fig:hellman-cover-wilson} shows such a protocol. It has no absorbing states and is fully manipulable by a strategic sender.  \Cref{tab:comparison} below summarizes the comparison with \citet{hellman1970learning} and \citet{wilson2014bounded} in detail.
\begin{figure}[ht]
\centering
\begin{tikzpicture}[->, >={Stealth[length=1.5mm, width=1.5mm]}, semithick, node distance=2.1cm]
  \node[circle, draw, minimum size=1cm] (S1) {$0$};
  \node[circle, draw, minimum size=1cm, right of=S1] (S2) {$0$};
  \node[circle, draw, minimum size=1cm, right of=S2] (S3) {$0$};
  \node[circle, draw, minimum size=1cm, right of=S3] (S4) {$0$};
  \node[circle, draw, minimum size=1cm, right of=S4] (S5) {$1$};

  \node[below=1mm of S1] {1};
  \node[below=1mm of S2] {2};
  \node[below=1mm of S3] {3};
  \node[below=1mm of S4] {4};
  \node[below=1mm of S5] {5};

  \coordinate[above=0.6cm of S3] (init);
  \draw[->] (init) -- (S3);

  \path
    (S1) edge[loop above]
      node {$l,(1-\epsilon)h$} (S1)
    (S1) edge[bend left=20]
      node[above] {$\epsilon h$} (S2)

    (S2) edge[bend left=20]
      node[below] {$l$} (S1)
    (S2) edge[bend left=20]
      node[above] {$h$} (S3)

    (S3) edge[bend left=20]
      node[below] {$l$} (S2)
    (S3) edge[bend left=20]
      node[above] {$h$} (S4)

    (S4) edge[bend left=20]
      node[below] {$l$} (S3)
    (S4) edge[bend left=20]
      node[above] {$h$} (S5)

    (S5) edge[bend left=20]
      node[below] {$\epsilon l$} (S4)
    (S5) edge[loop above]
      node {$h, (1-\epsilon)l$} (S5);
\end{tikzpicture}
\caption{A Hellman--Cover--Wilson protocol.}
\label{fig:hellman-cover-wilson}
\end{figure}

\begin{table}[ht]
\centering
\small
\begin{tabular}{|lll|}
\hline
\textbf{Feature} & \textbf{Hellman--Cover--Wilson} & \textbf{This paper} \\
Environment       & Single agent        & Strategic sender \\
Role of memory    & Learning only       & Learning + incentives \\
Protocol structure & Recurrent         & Absorbing + transient \\
Randomization & Transition out of extreme states & Transition into absorption \\
Manipulability & Yes        & No                     \\
Belief updating   & Continued updating & Eventual disengagement \\
Polarization     & Reversible                  & Irreversible \\

\hline
\end{tabular}
\caption{Comparison with non-strategic learning.}
\label{tab:comparison}
\end{table}
This comparison makes clear the central mechanism. Memory is not just a constraint on learning; it is also a strategic resource used to discipline the sender. The receiver’s problem is how to allocate scarce memory between two competing goals: learning about the state and deterring manipulation. Parsimonious protocols emerge as optimal by sacrificing some learning to improve incentives. Information disengagement need not reflect a behavioral bias; it serves as a commitment device for incentive provision. Randomization, in turn, partially substitutes for missing memory capacity and gives rise to persistent polarization even among receivers exposed to the same information.

\section{Structure of Optimal Protocols}\label{sec:protocol_design}

In this section, we characterize a central structural property of optimal decision protocols under the threat of manipulation.

\subsection{Optimality of Parsimonious Protocols}\label{sec:results}

\begin{definition}
    A protocol is \textbf{parsimonious} if it satisfies the following conditions:

    \textup{(i)} There are two absorbing memory states: one where \( L \) is played with probability $1$ and one where \( H \) is played with probability $1$.  

    \textup{(ii)} All other memory states are transient, with \( L \) played with probability $1$.

     A protocol is \textbf{simple} if there are two absorbing memory states and all other memory states are transient.
\end{definition}

A simple protocol may fail to be parsimonious because of its more flexible action rule.
With a simple protocol, the game ends in finite time almost surely regardless of the sender's strategy, and the receiver's payoff is independent of the sender's selection among best responses. These properties are established below. All omitted proofs are in the appendix. 

    \begin{lemma}\label{lem:payoff-equivalence}
For any protocol $\Pi$, the following statements hold.

\textup{(i)} The set of sender best responses is nonempty,
$\br(\Pi)\neq\emptyset$. Moreover, $\br(\Pi)$ contains a pure strategy
\(
\sigma:M\times\Theta\to\{0,1\}.
\)

\textup{(ii)} If $\Pi$ is simple, and hence in particular if it is
parsimonious, then the game ends in finite time almost surely regardless
of the sender's strategy. Moreover, the receiver's payoff is independent
of the sender's choice among best responses: for any
$\sigma,\sigma'\in\br(\Pi)$,
\(
U^R(\Pi,\sigma)=U^R(\Pi,\sigma').
\)
Thus, the receiver's payoff from $\Pi$ is well defined as
\[
U^R(\Pi)
:=
\sup_{\sigma\in\br(\Pi)}U^R(\Pi,\sigma)
=
\inf_{\sigma\in\br(\Pi)}U^R(\Pi,\sigma).
\]
\end{lemma}

The optimal design of a protocol must balance two competing goals: learning about the state of nature and deterring manipulation by the sender. With a parsimonious protocol, continuing to provide information until absorption is a best response for the sender, regardless of \(\theta\), so the receiver effectively controls the information flow. A parsimonious protocol is therefore non-manipulable, unlike the protocol in \Cref{fig:3-state}, and is effective for incentive provision. This benefit, however, comes at the cost of foregone learning. Our first main result shows that this form of incentive provision is nevertheless without loss of optimality.

\begin{theorem}\label{thm:parsimonious}
If a protocol $\Pi$ with $m$ states is not parsimonious, then there exists a parsimonious $\Pi'$ with at most $m$ memory states that weakly improves the receiver's payoff, $\sup_{\sigma \in \br(\Pi)} U^R(\Pi,\sigma)\leq U^R(\Pi')$.
\end{theorem}

\Cref{thm:parsimonious} and part (ii) of \Cref{lem:payoff-equivalence} immediately imply the following corollary, resolving the issue raised in \Cref{sec:expected payoff}.

\begin{corollary}
For any fixed memory set $M$ with $m\geq 2$ states, the receiver's optimal payoff is
\[
\sup_{\Pi}\sup_{\sigma\in\br(\Pi)} U^R(\Pi,\sigma)
=
\sup_{\Pi}\inf_{\sigma\in\br(\Pi)} U^R(\Pi,\sigma)
=
\sup_{\Pi \textup{ is parsimonious}} U^R(\Pi).
\]
\end{corollary}

The proof of \Cref{thm:parsimonious} proceeds in two main steps. The first is to transform a general protocol into a simple one with two absorbing states without changing the sender's incentives while guaranteeing the receiver's payoff. Thus, it is without loss of optimality to restrict attention to protocols with an absorption-transience structure.
This is summarized in \Cref{prop:algorithm-to-simple} in \Cref{sec:general-to-simple}. 

We then refine the action rule of a simple protocol into a parsimonious form. This is summarized in \Cref{prop:simple-to-parsimonious} in \Cref{sec:simple-to-parsimonious}. Parsimonious protocols differ from simple protocols in that their action rules are chosen to provide the strongest incentives for the sender. Actions at the transient states can, however, be perturbed slightly without affecting these incentives; see \Cref{footnote:parsimony} for the three-state example. Nonetheless, all optimal designs must essentially feature an absorbing structure (see \Cref{sec:non-simple}).

\subsection{From General to Simple Protocols}\label{sec:general-to-simple} 

\begin{proposition}\label{prop:algorithm-to-simple}
If a protocol $\Pi$ with $m$ memory states is not simple, then there exists a simple protocol $\Pi'$ with at most $m$ memory states such that $\sup_{\sigma \in \br(\Pi)}U^R(\Pi,\sigma)\leq U^R(\Pi')$.
\end{proposition}

We sketch the proof of \Cref{prop:algorithm-to-simple} and the main ideas below.
The proof is constructive. Starting from a protocol $\Pi^0=\Pi$ with
$m$ memory states, we iteratively construct a sequence
$\{\Pi^n\}$ as follows.

\begin{enumerate}
    \item If $\Pi^n$ has no absorbing state, then it contains a recurrent
    communicating class. Within such a class, consider a memory state
    $i$ with the highest value of $a(i)$. We obtain $\Pi^{n+1}$ by making
    state $i$ absorbing, while leaving all other components of the
    protocol unchanged. Intuitively, once the recurrent class is reached,
    the sender can optimally wait until such a most favorable state is
    reached, so making it absorbing does not hurt the receiver. This
    operation is formalized in
    \Cref{lem:recurrent-to-absorbing,cor:exist-absorbing}.

    \item If $\Pi^n$ has exactly one absorbing state, we can weakly improve
    it with a protocol $\Pi^{n+1}$ having two absorbing states. If there
    is another recurrent communicating class, the preceding argument
    produces the second absorbing state. Otherwise, all other states are
    transient. If none of them has a higher value of $a(i)$ than the
    existing absorbing state, the sender can optimally continue until
    absorption, and the resulting no-learning payoff can be replicated
    by a two-state parsimonious protocol. If some transient state has a
    higher value of $a(i)$, a state with the highest such value can be
    made absorbing as in the preceding step. This operation is formalized
    in \Cref{lem:1-to-2-absorbing}.

    \item If $\Pi^n$ has exactly two absorbing states but also has a
    distinct recurrent communicating class, the argument in the first
    step produces a third absorbing state without increasing the number
    of memory states.

    \item If $\Pi^n$ has more than two absorbing states, choose an
    intermediate absorbing state and replace every transition into it
    by an appropriate randomization over the two extreme absorbing states.
    This replicates its action probability while allowing that memory
    state to be eliminated. Hence we obtain a protocol $\Pi^{n+1}$ with
    one fewer memory state and one fewer absorbing state. This operation
    is formalized in \Cref{lem:reduce-absorbing-memory}.
\end{enumerate}

Whenever Steps 3 and 4 are repeated, Step 4 reduces the total number of memory states by one, so this cycle can occur only finitely many times. At termination, the protocol has exactly two absorbing states and all remaining states are transient; hence it is simple. The receiver's best attainable payoff across sender best responses is weakly improved at every step.

\subsection{From Simple to Parsimonious Protocols}\label{sec:simple-to-parsimonious}

A simple protocol does not necessarily compel the sender to continue sending signals until an absorbing state is reached, while a parsimonious protocol does. To reduce the receiver’s action rule to its parsimonious form, we must carefully examine the sender’s response in cases where a simple but non-parsimonious protocol induces a $\theta$-dependent stopping.   The central idea is that allowing the sender's private information
to determine whether she stops at a transient state provides no
essential benefit to the receiver. The receiver's action rule can be adjusted
until all of the sender's $\theta$-dependent stopping opportunities disappear, leaving a
protocol in which the sender continues providing information until one
of the two absorbing decision states is reached.

By (ii) of \Cref{lem:payoff-equivalence}, the receiver's payoff from any simple protocol $\Pi$ is independent of the sender's choice of best responses, which is simply written as $U^R(\Pi)$.  We can also pick any best response of the sender to compute the receiver's payoff, which greatly simplifies the remaining analysis.
\begin{proposition}\label{prop:simple-to-parsimonious}
    If a simple protocol $\Pi$ with $m$ memory states is not parsimonious, then there exists a parsimonious protocol $\Pi'$ with at most $m$ memory states such that $U^R(\Pi)\leq U^R(\Pi')$.
\end{proposition}

We sketch the proof of \Cref{prop:simple-to-parsimonious} and its main
ideas below. Fix a simple protocol $\Pi$ and a sender best response $\sigma$. We relabel the two absorbing states as $1$ and $m$,
with
\(
a(1)\leq a(m).
\)
Let $M_\theta$ denote the set of memory states at which the
sender stops when the state of nature is $\theta$, according to $\sigma$.  Our proof iteratively modifies $\Pi$ to eliminate
$\theta$-dependent stopping at transient memory states until $M_{\theta}$ is reduced to $\{1,m\}$.

\begin{enumerate}
    \item We first remove transient states whose roles are redundant.
    If the sender stops at a transient state under both states of nature,
    that state effectively acts as an additional absorbing state. We may
    therefore make it absorbing and then eliminate one of the resulting
    three absorbing states, as in step 4 in the proof of \Cref{prop:algorithm-to-simple}. If the sender continues at a
    transient state under both states of nature, the action prescribed
    there is never implemented under $\sigma$ and can be lowered to
    $a(1)$ without changing the sender's incentives or either player's
    payoff. These reductions are formalized in
    \Cref{lem:not-stopping-at-transient,lem:transient-state-action}.

    \item After these reductions, the only nontrivial transient states
    are those at which the sender stops under exactly one state of
    nature. These states form the symmetric difference
    \(
    M_H\triangle M_L.
    \)
    We may restrict the action probabilities at such states to the
    interval
    \(
    [a(1),a(m)].
    \)
    An action below $a(1)$ cannot induce optimal stopping, because
    continuing until absorption yields at least $a(1)$. A state
    prescribing an action above $a(m)$ can instead be made absorbing
    and subsequently eliminated. This reduction is formalized in
    \Cref{lem:bound-of-stopping-states}.

    \item We next eliminate the states in $M_H\triangle M_L$. Holding
    the transition rule $f$ and the selected sender best response $\sigma$ fixed, we
    modify the action rule $a$ on $M_H\triangle M_L$ to maximize
    the receiver's payoff subject to preserving the sender's
    best-response incentives and the bounds $[a(1),a(m)]$: 
    \begin{align}
    \sup_{\tilde{a}} \quad & U^R((f,g,\tilde{a}),\sigma) \notag \\
    \text{s.t.} \quad & \tilde{a}(i) = a(i) \text{ for all } i \notin M_H \triangle M_L, \\
    & \tilde{a}(i) \in [a(1),a(m)] \text{ for all } i \in M_H \triangle M_L  \\
    & \sigma \in \textup{br}(f,g,\tilde{a}). 
    \end{align}
    Both the
    receiver's objective and the sender's incentive constraints are
    linear in these action probabilities, so this is a linear program. At an extreme-point solution, either an action probability reaches
    one of the bounds or one of the sender's incentive constraints
    binds. If an incentive constraint binds, the sender is indifferent
    between stopping and continuing at some transient state under one
    state of nature, so her stopping behavior there can be made $\theta$-invariant. If a lower bound binds, a best
    response can be selected in which the sender continues at that state
    under both states of nature. If an upper bound binds, a best response
    can be selected in which she stops there under both states of nature.

    The memory state can then be treated using the logic of the first step. Thus, each
    iteration either reduces the number of memory states or strictly
    shrinks $M_H\triangle M_L$. This argument is formalized in
    \Cref{lem:reducing-stopping-states}. The procedure terminates after finitely many iterations.

    \item At termination, $M_H=M_L=\{1,m\}$ and every transient state prescribes action probability $a(1)$.
    Thus, the sender continues at every transient state under both
    states of nature. We then show that this action rule can be weakly improved by making $a(1)=0$ and $a(m)=1$. 
\end{enumerate}

\section{Payoff Characterization}\label{sec:optimal value}

Let the maximum and minimum likelihood ratios of the signals be
\begin{equation}\label{eqn:extreme-signals}
\bar\ell
:=
\max_{s\in S}\frac{\pi^H(s)}{\pi^L(s)}
\qquad \text{and} \qquad 
\underline\ell
:=
\min_{s\in S}\frac{\pi^H(s)}{\pi^L(s)}.
\end{equation}
For convenience, let $h$ and $l$ denote signals attaining
$\bar\ell$ and $\underline\ell$, respectively. Signals with the same
likelihood ratio are informationally equivalent and can be treated as
identical without loss. Define
$$
\gamma:=\frac{\bar\ell}{\underline\ell}
$$
as the \textit{spread} between the strongest and
weakest evidence for $\theta=H$. Thus $\gamma$ serves as a summary index of the
informativeness of the signal-generating process. Since $S$ is finite,
the signal distributions have full support, and
$\pi^H\not\equiv\pi^L$, we have
$
\gamma\in(1,\infty).
$ Write
\begin{equation}\label{eqn:more-less-likely-state}
\bar p := \max\{p,1-p\}
\qquad \text{and} \qquad 
\underline p := \min\{p,1-p\}
\end{equation}
as the prior probabilities of the
more likely and less likely states, respectively. Thus, $\bar p+\underline p=1$. Define
\[
\kappa
:=
\frac{\bar p}{\underline p}=
\max\left\{
\frac{p}{1-p},
\frac{1-p}{p}
\right\}
\]
as the \textit{bias of the prior}.

In what follows, we focus on the receiver's optimal payoffs. The treatment of the sender's optimal payoffs is in \Cref{sec:sender optimal value} of Online Appendix.

Given a parsimonious protocol $\Pi$, let $\mu_1^\theta(\Pi)$ and
$\mu_m^\theta(\Pi)$ denote the probabilities that the protocol is
eventually absorbed in states $1$ and $m$, respectively, when the state
of nature is $\theta$ and the sender continues providing information
until absorption. By \Cref{thm:parsimonious}, the receiver's optimization
problem can be written as
\begin{equation}
\tag{R}\label{eqn:receiver_opt}
\begin{aligned}
\sup_{\Pi}\quad
& p\mspace{1mu}\mu_m^H(\Pi)
  +(1-p)\mspace{1mu}\mu_1^L(\Pi)\\
\text{s.t.}\quad
& \Pi\text{ is parsimonious.}
\end{aligned}
\end{equation}

We solve this problem through a relaxation. \Cref{lem:spread} shows that
the absorption probabilities induced by any parsimonious protocol must
satisfy, whenever the relevant probabilities are interior,
\begin{equation}\label{eqn:relax-in-main-text}
\frac{\mu_m^H(\Pi)/\mu_1^H(\Pi)}
{\mu_m^L(\Pi)/\mu_1^L(\Pi)}
\leq
\gamma^{m-2}.
\end{equation}
Thus, $\gamma^{m-2}$ bounds the likelihood-ratio spread between the two
absorbing states and therefore captures an upper bound on the amount of learning permitted by $m$ memory states.

Replacing the requirement that the absorption probabilities be induced
by a parsimonious protocol with the necessary constraint
\eqref{eqn:relax-in-main-text} yields a relaxation of
\eqref{eqn:receiver_opt}, and hence an upper bound on the receiver's
optimal payoff. We then construct a sequence of parsimonious protocols
whose payoffs converge to the value of the relaxed problem, showing that
the upper bound is tight.

A necessary condition for nontrivial learning is that learning $\gamma^{m-2}$ can overcome the prior bias $\kappa$: \begin{equation*}
\gamma^{m-2}>\kappa.
\end{equation*} 
As \Cref{thm:supUR} shows, this condition is also sufficient for nontrivial learning, because \eqref{eqn:relax-in-main-text} binds at the relaxed optimum, so the upper bound for learning
$\gamma^{m-2}$ is fully exploited to overcome the prior bias. Moreover, when
$m\geq4$, every feasible parsimonious protocol with non-trivial learning satisfies
\eqref{eqn:relax-in-main-text} strictly. Hence the relaxed optimum cannot
be attained by any feasible protocol, although it can be approached
arbitrarily closely. This feature is demonstrated by
\Cref{fig:5-state-random-chain} and established formally in
\Cref{thm:supUR}.

\begin{theorem}\label{thm:supUR}
Fix $m\geq2$. The following statements hold.

\textup{(i)} The receiver's optimal payoff is
\begin{equation}\label{eqn:sup-payoff}
\sup_{\Pi}U^R(\Pi)
=
\begin{cases}
\displaystyle
\bar p
+
\underline p\,
\frac{
\left(\gamma^{(m-2)/2}-\sqrt{\kappa}\right)^2
}{
\gamma^{m-2}-1
},
&
\textup{if }\gamma^{m-2}>\kappa,
\\[1.3em]
\bar p,
&
\textup{otherwise.}
\end{cases}
\end{equation}

\textup{(ii)} If $\gamma^{m-2}>\kappa$ and $m\geq4$, the
supremum in \eqref{eqn:sup-payoff}  is not attained by any protocol.

\end{theorem}

    The interpretation of the payoff is intuitive. The term $\bar p$ is the receiver's
payoff from acting on the prior, while $\underline p=1-\bar p$ is the
maximal possible gain from learning. If $\gamma^{m-2}\leq\kappa$, the maximal amount of learning permitted by the receiver's memory
capacity is insufficient to overcome the prior bias, so it is
optimal to act on the prior alone. If $\gamma^{m-2}>\kappa$, learning
is valuable, and
\(
\frac{
\left(\gamma^{(m-2)/2}-\sqrt{\kappa}\right)^2
}{
\gamma^{m-2}-1
}
\)
is the learning efficiency of the optimal protocol relative to the full-information benchmark.

Comparing \Cref{thm:supUR} with the characterization in
\citet{hellman1970learning} shows that, in payoff terms, the cost of
providing incentives to the strategic sender is equivalent to one
memory state.
\begin{corollary}[The cost of incentive provision]\label{cor:incentive-cost}
The receiver's optimal value with $m\geq2$ memory states and a strategic
sender equals his optimal value in the non-strategic benchmark with
$m-1$ memory states.
\end{corollary}

If the receiver can specify the distribution of signals (e.g., the FDA often instructs a pharmaceutical company on what kinds of experiments to conduct), he always prefers more informative signals.
\begin{corollary}\label{cor:cs-receiver} The following comparative statics hold for the receiver's optimal payoff $\sup_{\Pi} U^{R}(\Pi):$

   \textup{(i)}  If $\gamma^{m-2}> \kappa$, then the receiver's optimal payoff is strictly increasing in the informativeness of signals $\gamma$ and the number of memory states $m$. As $\gamma \rightarrow \infty$ or $m \rightarrow \infty$, the optimal payoff goes to $1$. As $\gamma$ and $m$ decrease such that $\gamma^{m-2} \rightarrow \kappa$, the optimal payoff goes to $\bar{p}=\max\{p,1-p\}.$

   \textup{(ii)} If $\gamma^{m-2} \leq  \kappa$, then the receiver's optimal payoff is obtained by acting on the prior, which is independent of the informativeness of signals $\gamma$ and the number of memory states $m$.
\end{corollary}

For the symmetric binary case, where $p=\frac{1}{2}$, $S=\{h,l\}$, and $\pi^H(h)=\pi^L(l)=q\in (\frac{1}{2},1)$, we have $\bar{\ell}=\frac{q}{1-q}$, $\underline{\ell}=\frac{1-q}{q}$, $\gamma = (\frac{q}{1-q})^2$, and $\kappa=1$. Therefore, $\gamma^{m-2}>\kappa$ is equivalent to $m> 2$. With $m = 2$, \Cref{thm:supUR} implies that $\sup_{\Pi} U^{R}(\Pi)=\frac{1}{2}.$ The following implication of \Cref{thm:supUR} confirms that the protocols constructed in \Cref{sec:eg} are indeed optimal.
\begin{corollary}\label{cor: receiver-value-sym}
    For the symmetric binary case, if $m \geq 2$, then
    \[
\sup_{\Pi} U^{R}(\Pi)
=
\dfrac{q^{m-2}}{q^{m-2}+(1-q)^{m-2}}.\]
\end{corollary}

\section{Implications for Belief Updating and Behavior}\label{sec:behavior}

We now formalize the belief updating and behavior induced by asymptotically optimal parsimonious protocols: information disengagement, indecision near the decision point, and polarization conditional on a common signal sequence. We study the transition rules of asymptotically optimal parsimonious protocols. 

\subsection{Information Disengagement}\label{sec:Information Disengagement}

Under a parsimonious protocol, every non-absorbing memory state is
transient. The protocol
therefore reaches one of its two absorbing states in finite time almost
surely. Once an absorbing state is reached, the receiver's memory state
and action cease to respond to additional signals. In behavioral terms,
the receiver eventually disengages from further information.

This disengagement entails a genuine loss of learning. After reaching an absorbing state,
additional evidence can no longer correct an erroneous decision, even
when such evidence remains informative and costless. The
absorbing architecture nevertheless provides an incentive benefit. At
every transient state, a parsimonious protocol prescribes action $L$,
which is unfavorable to the sender. Continued information provision is 
optimal for the sender. 

The mechanism differs from standard optimal stopping in dynamic decision
problems. There is no discounting, deadline, information-acquisition
cost, or direct cost of waiting. Information disengagement is instead a
strategic commitment device: remaining perpetually responsive to new
information would leave the receiver continually exposed to manipulation
by the information provider. By becoming unresponsive after reaching an
absorbing state, the receiver limits this exposure and preserves control
over the information process.

Information disengagement is therefore not assumed to be a behavioral
bias. It is an equilibrium consequence of optimal protocol design under
asymmetric rationality. For example, a voter who makes up his mind and
stops attending to further media coverage, or a regulatory body that
closes an evidentiary record after reaching a decision, may be trading
off the value of additional information against continued susceptibility
to strategic influence. The decision may turn out to be mistaken ex
post, and further evidence may still be valuable, but commitment can
nevertheless be optimal ex ante.

The following corollary identifies the limiting probabilities
of the two final decisions. In particular, even along a sequence
approaching the receiver's optimal payoff, each absorbing state
is reached with a strictly positive limiting probability under
either state of nature.

\begin{corollary}\label{cor:interior-absorption}
Suppose $\gamma^{m-2}>\kappa$, and let
$\{\Pi^n\}_{n=1}^\infty$ be a sequence of $m$-state parsimonious
protocols satisfying
$U^R(\Pi^n)\to\sup_\Pi U^R(\Pi)$. Then
\begin{equation}\label{eqn:alphabeta-in-thm}
\begin{aligned}
\lim_{n\to\infty}\mu_m^H(\Pi^n)
&=
\alpha^*
:=
\frac{
\gamma^{m-2}
-\sqrt{\frac{1-p}{p}\gamma^{m-2}}
}{
\gamma^{m-2}-1
}
\in(0,1),
\\[1em]
\lim_{n\to\infty}\mu_1^L(\Pi^n)
&=
\beta^*
:=
\frac{
\gamma^{m-2}
-\sqrt{\frac{p}{1-p}\gamma^{m-2}}
}{
\gamma^{m-2}-1
}
\in(0,1).
\end{aligned}
\end{equation}
Consequently, all four absorption probabilities
$\mu_i^\theta(\Pi^n)$, for $i\in\{1,m\}$ and
$\theta\in\{H,L\}$, converge to limits in $(0,1)$.
\end{corollary}

Thus, approaching the optimal payoff does not eliminate mistaken
decisions. Under either state of nature, the receiver disengages
with a strictly positive limiting probability of having chosen
the wrong action.

\subsection{Indecision Near the Decision Point}\label{sec:Indecision}

Parsimonious protocols imply eventual information disengagement. The
five-state protocol in \Cref{sec:eg} suggests a further behavioral
implication: hesitation immediately before decision point. In that example,
states $2$ and $4$ are the only transient states from which the protocol
can enter the absorbing decision states $1$ and $5$, respectively, and
the corresponding one-step absorption probabilities become arbitrarily
small. The result below shows that this structure is a general feature
of asymptotically optimal protocols.

For a parsimonious protocol $\Pi$, consider the Markov chain over memory
states induced by continued information provision.
The chain is reducible and hence the ordinary stationary distributions do not describe its behavior before
absorption. Let $\nu^\theta_i(\Pi)$ denote the relative hitting frequency of a transient memory state $i$ when the state of nature is $\theta$.  We label the transient states ${2,\ldots,m-1}$ in the order of their likelihood ratios $\frac{\nu^H_i(\Pi)}{\nu^L_i(\Pi)}$: $\frac{\nu^H_2(\Pi)}{\nu^L_2(\Pi)}\leq \dots \leq \frac{\nu^H_{m-1}(\Pi)}{\nu^L_{m-1}(\Pi)}.$
Thus, states $2$ and $m-1$ are the transient states most indicative of
$L$ and $H$, respectively.\footnote{See \Cref{sec: proof with modified chain} for mathematical details. Ties are broken arbitrarily. Along
an asymptotically optimal sequence, ties occur only finitely often; see
\Cref{rem:equal-likelihood-transient-states}.}

Recall that $l$ and $h$ attain the lowest and highest signal likelihood
ratios, respectively, in \eqref{eqn:extreme-signals}. We focus on the
nontrivial learning region
$\gamma^{m-2}>\kappa$
and on $m\geq4$. Let $\tau^\theta(\Pi)$ be the random time to absorption induced by a parsimonious protocol $\Pi$ when the state of nature is $\theta$. Then $m_{\tau^\theta(\Pi)}\in \{1,m\}$, and $m_{\tau^\theta(\Pi)-1}$ denotes the memory state immediately before absorption.\footnote{If $\tau^\theta(\Pi)=0$, absorption occurs at
initialization. For notational convenience, set
$(m_{-1},s_{-1})=(2,l)$ when $m_0=1$, and
$(m_{-1},s_{-1})=(m-1,h)$ when $m_0=m$.
These assignments are purely notational. Along the
asymptotically optimal sequences considered below, initial
absorption has vanishing probability conditional on either
absorbing state, so this convention does not affect the
stated limits.} Our goal is to understand the memory states right before absorption and the signals that trigger these transitions.

\begin{theorem}\label{thm:behavioral_implication}
Suppose $\gamma^{m-2}>\kappa$ and $m\geq 4$, and let $\{\Pi^n\}_{n=1}^\infty$ be a sequence of parsimonious protocols that asymptotically attains the receiver's optimal payoff, 
$\lim\limits_{n \rightarrow \infty} U^R(\Pi^n) = \sup_\Pi U^R(\Pi)$. 
Then the following statements hold as $n\rightarrow\infty$.

\textup{(i)} The probability that an absorbing state is reached from its corresponding extreme transient state upon receiving an extreme signal converges to one for each 
$\theta \in \Theta$:
\begin{equation*}
    \textup{P}^\theta(m_{\tau^\theta(\Pi^n)-1}=2,\, s_{\tau^\theta(\Pi^n)-1}=l \mid m_{\tau^\theta(\Pi^n)}=1) \rightarrow 1,
\end{equation*}
\begin{equation*}
    \textup{P}^\theta(m_{\tau^\theta(\Pi^n)-1}=m-1,\, s_{\tau^\theta(\Pi^n)-1}=h \mid m_{\tau^\theta(\Pi^n)}=m) \rightarrow 1.
\end{equation*}

\textup{(ii)} The transition to an absorbing state is stochastic, and the corresponding transition probability vanishes:
\begin{equation*}
    f^n(i,s)(\{1,m\}) \rightarrow 0
\end{equation*}
for all $i \in \{2,..., m-1\}$ and $s\in S$.
\end{theorem}

Part \textup{(i)} shows that commitment becomes increasingly selective:
the protocol enters the $L$-decision state only from the most
$L$-indicative transient state upon receiving the most
$L$-indicative signal, and analogously for the $H$-decision state.
Part \textup{(ii)} shows that even in these configurations, the
probability of immediate commitment becomes arbitrarily small.

We refer to this vanishing one-step hazard of irreversible commitment as
\textit{indecision near the decision point}. It does not imply that the
receiver fails to decide. Every protocol in the sequence reaches an
absorbing state almost surely, and
\Cref{cor:interior-absorption} implies that both final decisions remain
possible with probabilities bounded away from zero. Rather, the
receiver reaches a final decision through increasingly rare
opportunities for commitment. As the five-state example illustrates,
this randomization allows the protocol to use repeated observations
before committing, partially substituting for unavailable memory
capacity. The fact that the decision maker reacts only to extreme signals, however, is unsurprising in light of \citet{hellman1970learning} and \citet{wilson2014bounded}.

\subsection{Polarization}\label{sec:Polarization}

Stochastic transitions allow identical receivers with identical information to reach different absorbing states and hence take different actions. We now establish a strong form of polarization: the probability of such disagreement remains bounded away from zero along any asymptotically optimal sequence, even as the probability of transitioning into absorbing states vanishes. Once absorbed, the receivers cease to respond to further signals. Thus, within the committed protocol, the same commitment that protects them against manipulation may also prevent subsequent common evidence from reconciling their decisions.

    Fix a parsimonious protocol $\Pi$ and a state of nature $\theta\in\{H,L\}$. Let \[ \mathbf{s}=(s_0,s_1,\ldots)\sim(\pi^\theta)^\infty \] be an infinite signal sequence. Conditional on $\mathbf{s}$, consider two independent copies of $\Pi$, with independent initial-state and transition randomization, both driven by the same signal sequence $\mathbf{s}$. Let $K_r^{\theta}(\Pi,\mathbf{s})\in\{1,m\}$ denote the absorbing state reached by copy $r\in\{1,2\}$. 

\begin{definition}\label{def:polarization} 
The polarization probability conditional on the signal sequence $\mathbf{s}$ and the state of nature $\theta$ is \[ \begin{aligned} \lambda^\theta(\Pi\mid\mathbf{s}) :={}& {\Pr}^{\theta}\bigl( K_1^{\theta}(\Pi,\mathbf{s})=1,\, K_2^{\theta}(\Pi,\mathbf{s})=m \mid\mathbf{s} \bigr)\\ &+ {\Pr}^{\theta}\bigl( K_1^{\theta}(\Pi,\mathbf{s})=m,\, K_2^{\theta}(\Pi,\mathbf{s})=1 \mid\mathbf{s} \bigr), \end{aligned} \] where the probability is taken over the two copies' independent random transitions. The polarization probability of $\Pi$ in state of nature $\theta$ is \[ \lambda^\theta(\Pi) := \mathbb{E}_{\mathbf{s}\sim(\pi^\theta)^\infty} \left[ \lambda^\theta(\Pi\mid\mathbf{s}) \right]. \] 
\end{definition} 

Thus, our notion of polarization measures the residual randomness in the terminal
decision after conditioning on the entire signal sequence.

\begin{theorem}\label{thm:polarization}
    Suppose $m\geq4$ and $\gamma^{m-2}>\kappa$. Let $\{\Pi^n\}_{n=1}^\infty$ be any sequence of parsimonious $m$-state protocols that achieves the receiver's optimal payoff, i.e., $\lim\limits_{n \rightarrow \infty} U^R(\Pi^n) = \sup_\Pi U^R(\Pi)$. Then, for each $\theta\in \{H,L\}$,
    \begin{equation*}
       \liminf_{n\rightarrow \infty}\lambda^\theta(\Pi^n)>0.
    \end{equation*}
\end{theorem}

The theorem establishes a strong form of polarization:
receivers exposed to the same external information and following the
same asymptotically optimal protocol reach different final decisions
with a probability that does not vanish. It does not assert that both
absorbing states remain possible conditional on every signal sequence.

The restrictions in \Cref{thm:polarization} cannot be dropped in
general. When $\gamma^{m-2}\leq\kappa$, acting on the prior is optimal
and can be implemented by a protocol with no polarization. When $m=3$,
an optimal protocol may have deterministic transitions, as in the
binary-symmetric example in \Cref{sec:eg}, and therefore need not
generate polarization.

\section{Discussion of Assumptions}\label{sec:discussion}

\textbf{The Sender's Stopping as a Signal.}
We have specified that, when the sender stops the game, the receiver
implements the action prescribed at his current memory state. Allowing
the receiver to update this state in response to stopping does not
change the analysis. Suppose that, when the sender stops the information
flow in a non-absorbing memory state $i$, the protocol first transitions
according to
\[
f(i,\textup{``stop''})\in\Delta(M),
\]
and then immediately implements the action prescribed at the resulting
state. The same distribution over terminal actions can be implemented
without a stop-triggered transition by setting
\[
a'(i)=\sum_{j=1}^{m}f(i,\textup{``stop''})(j)a(j)
\]
at every non-absorbing state $i$, while leaving the initial distribution,
signal-contingent transitions, and action rule at absorbing states
unchanged. The resulting protocol preserves both players' expected payoffs for every sender strategy,
and hence the sender's best responses.

Thus, allowing explicit transitions in response to stopping does not
expand the set of attainable payoffs. In particular, responses to
stopping, including punishments, can already be implemented through the
action rule. We omit these redundant transitions to keep the formulation
simple and facilitate comparison with the non-strategic benchmark.

\textbf{Observability of Protocols.}
We have assumed that the sender observes the receiver’s protocol. If, instead, the receiver could keep the protocol secret, his strategy space would consist of randomizing over automata. For a fixed set of memory states, such randomization cannot, in general, be replicated by stochastic transitions within a single automaton. The resulting strategy space is considerably more difficult to analyze, and the interpretation of automata as decision protocols becomes less transparent. The distinction between the two cases is analogous to that between Stackelberg and Nash equilibrium.

\textbf{The Sender's Information.}
We assume the sender observes the receiver's current memory
state and uses a stationary Markov strategy
$\sigma:M\times\Theta\rightarrow[0,1]$. More generally, her strategy
can be written as $\sigma:\mathcal{I}\times\Theta\rightarrow[0,1]$,
where $\mathcal{I}$ is the collection of her information sets.
Two alternative information structures are:

\begin{enumerate}
    \item \textit{Signal history only.}
    $\mathcal{I}=\bigcup_{t\geq0}S^t$, where
    $S^0=\{\varnothing\}$. At the beginning of period $t$, the sender
    observes the complete signal history $(s_0,\ldots,s_{t-1})$,
    but not the receiver's memory states. At $t=0$, she observes
    the empty history.

    \item \textit{Complete history.}
    $\mathcal{I}=\bigcup_{t\geq0}(S^t\times M^{t+1})$.
    At the beginning of period $t$, the sender observes all past
    signals and all past and current memory states. At $t=0$,
    she observes $m_0$; at $t\geq1$, she observes
    $(s_0,\ldots,s_{t-1},m_0,\ldots,m_t)$.
\end{enumerate}

Conditional on $\theta$, the current memory state is a sufficient statistic for
the sender's stopping problem. Thus, complete-history observation
allows additional history-dependent strategies but does not increase
her optimal payoff: a stationary Markov best response suffices.

For parsimonious protocols, continued information provision until
absorption is a best response under all three information structures,
and the receiver's payoff is independent of the sender's choice among
best responses. These protocols therefore deliver the same payoffs
across the three scenarios, making the baseline a useful ``omniscient''
benchmark for payoff and behavioral analysis.

The signal-history-only scenario is different, since the receiver
could potentially benefit from hiding his current memory state.  The case is more complex because the sender must
base her stopping decision on a distribution over current memory
states, inferred from the signal history and conditional on the game
remaining active.
Nevertheless, \Cref{cor:incentive-cost} shows that our parsimonious
constructions with $m$ memory states approach the non-strategic
benchmark with $m-1$ states. Their robustness across information
structures therefore provides a payoff guarantee when the
sender observes only the signal history. 

\textbf{The Receiver's Commitment.}
We have also assumed that the receiver can commit to the protocol rather than modify it freely as the game unfolds. This assumption is less consequential when any modification is publicly observable, but private modifications introduce new strategic complications. More fundamentally, allowing private adjustments raises a conceptual issue of time consistency in the presence of imperfect recall; see \citet{piccione1997interpretation} and the references therein.  

\textbf{Discounting.}
Finally, we assume no discounting, exogenous stopping, or related time frictions, so the non-strategic learning component of the model is closer to \citet{hellman1970learning} than to \citet{wilson2014bounded}. Such frictions can restore continuity to the stopping problem,\footnote{See \Cref{footnote:hellman-wilson} for further discussion.} but they also create an independent incentive to stop early in the strategic learning problem and substantially complicate the analysis. By abstracting from these forces, our model isolates stopping as an equilibrium response to manipulation and counter-manipulation and provides the simplest formulation of a nontrivial strategic problem, at the cost of payoff discontinuity. A full analysis of the discounting version is beyond the scope of this paper and is left for future work.

\newpage

\appendix

\allowdisplaybreaks

\begin{center}
    \Large{\textbf{Appendix: Omitted Proofs}}
\end{center}

\numberwithin{theorem}{section}
\numberwithin{corollary}{section}
\numberwithin{claim}{section}
\numberwithin{lemma}{section}
\numberwithin{proposition}{section}
\numberwithin{definition}{section}
\numberwithin{remark}{section}
\numberwithin{example}{section}

\section{Proof of \Cref{lem:payoff-equivalence}}
    (i) is a known result for finite Markov problem without discounting; one argument is given by \citet[Chapter III]{dynkin1969markov}; the subtlety is that the sender’s payoff may be discontinuous when a strategy approaches one that never stops. 

    To prove (ii), let $\bar{a}_{\theta}(\Pi,\sigma)$ be the total probability that the high action is induced from the receiver conditional on state $\theta$ if the strategy profile $(\Pi,\sigma)$ is carried out. If $\Pi$ is simple, then the game ends in finite time almost surely for any $(\Pi,\sigma)$. It follows from $\sigma,\sigma' \in \textup{br}(\Pi)$ that $\bar{a}_{\theta}(\Pi,\sigma)=\bar{a}_{\theta}(\Pi,\sigma')$ and hence
\begin{equation*}
        U^R(\Pi, \sigma)= p\bar{a}_H(\Pi, \sigma)+(1-p)(1-\bar{a}_L(\Pi, \sigma))=U^R(\Pi, \sigma')
\end{equation*} as desired. 

\section{Proof of \Cref{thm:parsimonious}}\label{proof:thm:parsimonious}

\subsection{Proof of \Cref{prop:algorithm-to-simple}}\label{proof:general-to-simple}

\begin{lemma}\label{lem:stopping-exists}
For any protocol $\Pi$, there exists a pure strategy
$\sigma\in\br(\Pi)$ such that the game ends in finite time
almost surely.
\end{lemma}

\begin{proof}
By \Cref{lem:payoff-equivalence}\textup{(i)}, choose a pure
best response $\sigma^0\in\br(\Pi)$. Consider the recurrent
communicating classes of the Markov chain induced by continued
information provision. Full support of the signal distributions
implies that these classes are the same under both states of
nature.

For each recurrent class $C$, choose
$i_C\in\arg\max_{j\in C}a(j)$. Modify $\sigma^0$ within $C$
so that the sender continues until reaching $i_C$ and then stops,
under either state of nature. If $C$ is a singleton absorbing
state, the game already ends upon entry. Leave $\sigma^0$
unchanged at all transient states, and denote the resulting
pure strategy by $\sigma$.

Since $C$ is finite and irreducible, $i_C$ is reached in finite
time almost surely under continued information provision.
The modified strategy therefore yields payoff $a(i_C)$ from
every state in $C$. No strategy can yield more after entering
$C$: the class is closed, every terminal action probability
there is at most $a(i_C)$, and nontermination yields zero.
Thus, the modification weakly increases the sender's payoff.
Because $\sigma^0$ is already a best response, $\sigma$ must
also be a best response.

Finally, the process either stops at a transient state or enters
a recurrent class in finite time almost surely. Under $\sigma$,
it also terminates in finite time almost surely after entering
any recurrent class. Hence the game ends in finite time almost
surely.
\end{proof}

\begin{lemma} \label{lem:receiver-optimal-stops}
    For any protocol $\Pi$ and $\sigma,\sigma'\in \br(\Pi)$, if under $\sigma$ the game ends almost surely then $U^R(\Pi,\sigma)\geq U^R(\Pi,\sigma')$.
\end{lemma}

\begin{proof}
    Let $\bar{a}_{\theta}(\Pi,\sigma)$ be the total probability that the high action is induced from the receiver conditional on state $\theta$ if the strategy profile $(\Pi,\sigma)$ is carried out (and similarly for $\sigma'$). Let $\bar{b}_{\theta}(\Pi,\sigma')$ be the total probability that the game does not end in finite time conditional on state $\theta$ if the strategy profile $(\Pi,\sigma')$ is carried out. Then
    \begin{equation*}
    \begin{aligned}
        U^R(\Pi, \sigma)-U^R(\Pi,\sigma')&= p\bar{a}_H(\Pi, \sigma)+(1-p)(1-\bar{a}_L(\Pi, \sigma))\\
        &-p\bar{a}_H(\Pi, \sigma')-(1-p)(1-\bar{a}_L(\Pi, \sigma')-\bar{b}_L(\Pi, \sigma'))\\
        &=(1-p)\bar{b}_L(\Pi, \sigma')\\
        &\geq 0,
    \end{aligned}
    \end{equation*}
where the second equality follows because $\sigma,\sigma'\in \br(\Pi)$ and hence $\bar{a}_\theta(\Pi, \sigma)=\bar{a}_\theta(\Pi, \sigma')$ for each $\theta\in \{H,L\}$.
\end{proof}

\begin{corollary}\label{cor:pure-strategy}
    For any protocol $\Pi$, for $\sigma'\in \br(\Pi)$, if under $\sigma'$ the game ends in finite time almost surely, then $U^R(\Pi,\sigma')=\sup_{\sigma\in \br(\Pi)} U^R(\Pi,\sigma)$. Furthermore, such $\sigma'$ can be taken as a pure strategy.
\end{corollary}

\begin{lemma}\label{lem:recurrent-to-absorbing}
    Suppose a protocol \( \Pi \) has \( m \) memory states, including \( n \geq 0 \) absorbing states and a distinct recurrent communicating class. Then, there exists a protocol \( \Pi' \) with \( m \) memory states, where \( n+1 \) of them are absorbing states, such that \( \sup_{\sigma\in \br(\Pi)} U^R(\Pi,\sigma) = \sup_{\sigma'\in \br(\Pi')}U^R(\Pi',\sigma') \).
\end{lemma}
\begin{proof}
    Let state $i$ be a state with the highest \( a(i) \) in this recurrent communicating class of $\Pi=(f,g,a)$.  By converting \( i \) in $\Pi$ into an absorbing state and keeping everything else unchanged, we obtain a new automaton $\Pi'=(f',g,a),$ where 
    \begin{equation}\label{eqn:make-absorbing}
    f'(j, s)(j') =
    \begin{cases} 
        1, & \text{if } j'=j=i,  \\
        0,  & \text{if } j'\neq j=i,  \\
        f(j, s)(j'), & \text{otherwise.}
    \end{cases}
\end{equation}
Take any $\sigma\in \br(\Pi)$ such that under $\sigma$ the game ends in finite time almost surely. The existence of such a $\sigma$ is guaranteed by \Cref{lem:stopping-exists}. If the recurrent communicating class is reached under $(\Pi,\sigma)$,  then $\sigma$ must have action $H$ being implemented with probability $a(i)$. Since the game ends in finite time almost surely under $\sigma$, $\sigma$ must have action $L$ being implemented with probability $1-a(i)$ in this recurrent communicating class. This means that $\sigma\in \br(\Pi')$ and $U^R(\Pi,\sigma)=U^R(\Pi',\sigma)$. By \Cref{lem:receiver-optimal-stops}, $\sup_{\sigma\in \br(\Pi)} U^R(\Pi,\sigma) = \sup_{\sigma'\in \br(\Pi')}U^R(\Pi',\sigma')$.
\end{proof}

Since a Markov process without an absorbing state must contain a recurrent communicating class, \Cref{lem:recurrent-to-absorbing} implies the following:
\begin{corollary}\label{cor:exist-absorbing}
    If \( \Pi \) is a protocol with \( m \) memory states and no absorbing states, then there exists a protocol \( \Pi' \) with \( m \) memory states, including one absorbing state, such that \( \sup_{\sigma\in \br(\Pi)} U^R(\Pi,\sigma) = \sup_{\sigma'\in \br(\Pi')}U^R(\Pi',\sigma') \).
\end{corollary}

\begin{lemma}\label{lem:1-to-2-absorbing}
    If \( \Pi \) is a protocol with \( m \) memory states, including \(1\) absorbing state, then there exists a protocol \( \Pi' \) with \( m \geq 2 \) states, including \(2\) absorbing states, such that \( \sup_{\sigma\in \br(\Pi)} U^R(\Pi,\sigma) \leq \sup_{\sigma'\in \br(\Pi')}U^R(\Pi',\sigma') \).
\end{lemma}

\begin{proof}
Call the absorbing state in \( \Pi\) state \( 1 \). Suppose there is a distinct recurrent communicating class. Since this class cannot be another absorbing state, it must contain more than one memory state. The result then follows from \Cref{lem:recurrent-to-absorbing}.

If there are no distinct recurrent communicating classes, every memory state other than \( 1 \) is transient and eventually reaches state \( 1 \) with probability \( 1 \). There are two cases to consider.

{First}, if for every state \( i \neq 1 \), \( a(i) \leq a(1) \), then one of the sender's best responses is to keep generating signals until the memory state transitions to \( 1 \), and the receiver cannot get a higher payoff against any other sender best response. Therefore, 
\begin{equation*}
    \sup_{\sigma \in \br(\Pi)}U^R(\Pi,\sigma) = a(1)p + (1 - a(1))(1 - p) \leq \max\{p, 1 - p\}.
\end{equation*}
However, a parsimonious protocol \( \Pi' \) with two absorbing states, where the initial distribution assigns probability \( 1 \) to one of the two absorbing states, achieves the upper bound.

{Second}, if there exists a memory state \( i \) such that \( a(i) > a(1) \), consider the memory state with the highest \( a \), which we still denote as \( i \). As in the proof of \Cref{lem:recurrent-to-absorbing}, defining a new protocol $\Pi'$ from $\Pi$ by making \( i \) a new absorbing state. Any sender best response to $\Pi$ remains a best response to $\Pi'$ and they lead to the same receiver payoff. The result follows immediately. 
\end{proof}

\begin{lemma}\label{lem:reduce-absorbing-memory}
If a protocol \( \Pi \) on $M$ has \( m \) memory states including \( n > 2 \) absorbing states, then there exists a protocol \( \Pi' \)  on $M' \subset M$ with \( m-1 \) memory states including \( n-1 \) absorbing states such that \( \sup_{\sigma\in \br(\Pi)} U^R(\Pi,\sigma) = \sup_{\sigma'\in \br(\Pi')}U^R(\Pi',\sigma') \). In particular, if \( \sigma  \in \br(\Pi) \), then $\sigma_{M'}\in \br(\Pi')$, where $\sigma_{M'}$ is the restriction of \( \sigma \) to \( M' \). 
\end{lemma}

\begin{proof}
    For \( \Pi = (f, g, a) \), relabel two absorbing states with the lowest and highest probabilities \( a \) as \( 1 \) and \( m \), respectively, and denote these probabilities by \( a_{\min} \) and \( a_{\max} \). For an absorbing state \( k \neq 1, m \), we eliminate it from \(\Pi\) as follows. Every transition from \( j \) to \( k \) upon signal \( s \) with probability \( f(j, s)(k) \) in \(\Pi\) is reassigned to states \( 1 \) and \( m \), with proportions 
\begin{equation*}
\rho=\frac{a_{\max}-a(k)}{a_{\max}-a_{\min}}
\qquad\text{and}\quad
1-\rho=\frac{a(k)-a_{\min}}{a_{\max}-a_{\min}}.
\end{equation*}
respectively (if \( a_{\min} = a_{\max} \), we set \( \rho = \frac{1}{2} \)).

We also modify the initial distribution \( g \) accordingly: the probability assigned to \( k \) is reassigned to \( 1 \) and \( m \) with proportions \( \rho \) and \( 1 - \rho \), respectively.

This modification yields a new protocol \( \Pi' = (f', g', a') \) on \( M' = M \setminus \{k\} \) with \( m - 1 \) memory states, where:
\begin{equation*}
\begin{alignedat}{2}
    f'(j, s)(i) &= 
    \begin{cases} 
        f(j, s)(i) + \rho f(j, s)(k),       &\quad \text{if } i = 1, \\
        f(j, s)(i) + (1 - \rho) f(j, s)(k), &\quad \text{if } i = m, \\
        f(j, s)(i),                         &\quad \text{otherwise},
    \end{cases} \\
    g'(i) &= 
    \begin{cases} 
        g(i) + \rho g(k),       &\quad \text{if } i = 1, \\
        g(i) + (1 - \rho) g(k), &\quad \text{if } i = m, \\
        g(i),                   &\quad \text{otherwise}.
    \end{cases}
\end{alignedat}
\end{equation*}
Additionally, \( a' = a \) on \( M' \).

A best response \( \sigma' \) to \( \Pi' \) can be derived from a best response \( \sigma \) to \( \Pi \) by eliminating the response in memory state \( k \). The modification preserves both players' payoffs. The result then follows immediately.
\end{proof}

\begin{proof}[\textbf{Proof of \Cref{prop:algorithm-to-simple}}]

    Starting with a protocol $\Pi^0=\Pi$ with $m$ memory states, we iteratively apply the following algorithm to obtain a sequence $\{\Pi^n\}$:
     \begin{enumerate}
         \item If $\Pi^n$ has no absorbing state, apply \Cref{cor:exist-absorbing} to obtain an automaton $\Pi^{n+1}$ with the same number of memory states and at least one absorbing state.
         \item If $\Pi^n$ has exactly one absorbing state, apply \Cref{lem:1-to-2-absorbing} to obtain a protocol $\Pi^{n+1}$ with the same number of memory states and two absorbing states, or with just two absorbing states.
         \item If $\Pi^n$ has two absorbing states and a distinct recurrent communicating class, apply \Cref{lem:recurrent-to-absorbing} to obtain a protocol $\Pi^{n+1}$ with the same number of memory states and three absorbing states. 
         \item If $\Pi^{n}$ has more than two absorbing states, apply \Cref{lem:reduce-absorbing-memory} to obtain a protocol $\Pi^{n+1}$ with fewer memory states and fewer absorbing states. 
    \end{enumerate}

    Each time Step 4 is applied, the total number of memory states decreases by 1. Thus, the process terminates in a finite number of steps, resulting in a simple protocol. In each step, the receiver's supremum payoff never decreases. 
\end{proof}

\subsection{Proof of \Cref{prop:simple-to-parsimonious}}\label{proof:prop:simple-to-parsimonious}
\begin{lemma}\label{lem:not-stopping-at-transient}
     Let $\Pi$ be any simple protocol. Then there exists a simple protocol $\Pi'$, defined on a subset of the memory states of $\Pi$, such that $U^R(\Pi) = U^R(\Pi')$, and for which there exists a strategy $\sigma \in \br(\Pi')$ such that no transient state $i$ is one where $\sigma$ stops in both states of nature.
\end{lemma}
\begin{proof}
   Suppose \( \Pi \) does not satisfy this property. Then, for any \( \sigma \in \text{br}(\Pi) \) there exists a transient state \( i \) such that $\sigma$ stops in \( i \) in both states of nature. As in \Cref{eqn:make-absorbing} in the proof of \Cref{lem:recurrent-to-absorbing}, we define \( \Pi'' \) by making \( i \) a new absorbing state in \( \Pi \), removing transitions out of \( i \) while leaving the rest of \(\Pi\) unchanged.  By construction, \( \sigma \in \br(\Pi'') \) and \( U^R(\Pi) = U^R(\Pi'') \). Using \Cref{lem:reduce-absorbing-memory}, we can remove one of the three absorbing memory states from \( \Pi'' \) to obtain \( \Pi' \), where the restriction of \( \sigma \) to the memory states of \( \Pi' \) is a best response to \( \Pi' \), and \( U^R(\Pi') = U^R(\Pi'') = U^R(\Pi) \). The desired result follows from iterating this argument.
\end{proof}

\Cref{lem:not-stopping-at-transient} does not reduce simple protocols that have a transient memory state where the sender optimally stops under at most one of the two states of nature.

\begin{remark}\label{notation-1-m}
Without loss of generality, assume that for any simple protocol defined on a \textit{subset} of $\{1, \ldots, m\}$, the two absorbing memory states are $1$ and $m$, with \( a(1) \leq a(m) \). 
\end{remark}

\begin{lemma}\label{lem:transient-state-action}
    Let $\Pi = (f, g, a)$ be a simple protocol, and let $\sigma \in \br(\Pi)$ be a best response. For any transient state $i$ where $\sigma$ does not stop under either state of nature, define a modified action rule $a'$ by setting $a'(i) = a(1)$ and $a'(j) = a(j)$ for all other $j$. Then $\sigma$ remains a best response to the modified protocol $\Pi' = (f, g, a')$, and the receiver’s payoff satisfies $U^R(\Pi) = U^R(\Pi')$.
\end{lemma}

\begin{proof}
    Consider the sender’s strategy $\sigma$ under $\Pi'$. Since $\sigma$ does not stop in memory state $i$ under either state of nature, the change from $a$ to $a'$ does not affect the optimality of $\sigma$ in any memory state $j \neq i$. Moreover, since $a'(i) = a(1) \leq a(m) = a'(m)$, continuing to generate signals until reaching an absorbing state is at least as good as stopping in memory state $i$. Therefore, $\sigma$ remains a best response to $\Pi'$. Under both $\Pi$ and $\Pi'$, the total probability of taking action $H$,  $\bar{a}_\theta$, is unchanged, so the receiver’s payoff $U^R$ is unchanged as well.
\end{proof}

By \Cref{lem:not-stopping-at-transient} and \Cref{lem:transient-state-action}, we shall focus on simple protocols and sender's best responses that satisfy the following property:

\begin{definition}\label{def:simple-profile}
    We say that the profile \( (\Pi, \sigma) \) is a \textbf{simple profile} if \( \Pi \) is a simple protocol and \( \sigma \in \br(\Pi) \) satisfies the following conditions:

    \textup{(i)} There is no transient state \( i \) of \( \Pi \) where \( \sigma \) stops under both states of nature.

    \textup{(ii)} For any transient state \( i \) in which \( \sigma \) does not stop under either state of nature, we have \( a(i) = a(1) \). 

    Given a simple profile $(\Pi, \sigma)$, let $M_{\theta}$ denote the set of memory states in which $\sigma$ stops, conditional on the state of nature being $\theta$. We refer to $M_{\theta}$ as the \textbf{stopping set} for  $\theta$.
\end{definition}

\begin{remark}
   By \Cref{notation-1-m} and \Cref{def:simple-profile}, we have that for all simple profiles, $M_H \cap M_L = \{1, m\}$, and $a(i) = a(1)$ for all $i \in M \setminus (M_H \cup M_L)$. It is also without loss of consider pure strategy $\sigma$ in a simple profile.
\end{remark}

\begin{lemma}\label{lem:bound-of-stopping-states}
    For any simple protocol $\Pi$, there exists a simple profile $(\Pi', \sigma')$ such that $a'(i) \in [a'(1), a'(m)]$ for all $i \in M'_H \triangle M'_L$, and $U^R(\Pi) \leq U^R(\Pi')$.
\end{lemma}

\begin{proof}
    Let $\sigma$ be any best response to $\Pi$. Suppose the profile $(\Pi, \sigma)$ does not satisfy the stated property. Then there exists some $i \in M_H \triangle M_L$ such that $a(i) < a(1)$ or $a(i) > a(m)$. 

    If $a(i) < a(1)$, then the sender is strictly better off continuing to generate signals until reaching an absorbing state. Hence, $\sigma$ is not a best response to $\Pi$, a contradiction.

    If $a(i) > a(m)$, consider a memory state---possibly $i$ itself---in which action $H$ is played with the highest probability in $\Pi$. Then the sender has a best response in which she stops in that state under both states of nature. By \Cref{lem:not-stopping-at-transient}, there exists a simple protocol $\Pi'$ with strictly fewer memory states and $U^R(\Pi) \leq U^R(\Pi')$. We then check whether the condition $a(i) > a(m)$ still holds in $\Pi'$, and continue eliminating such states iteratively until $i$ is removed.

    The  result follows by iterating this process until all such $i$ are eliminated.
\end{proof}

\begin{lemma}\label{lem:reducing-stopping-states}
    Suppose $(\Pi,\sigma)$ is a simple profile with stopping sets
$M_H\neq M_L$ and $a(i)\in[a(1),a(m)]$ for every
$i\in M_H\triangle M_L$. Then there exists  a simple profile $(\Pi',\sigma')$ with stopping sets $M'_H$ and $M'_L$ such that the following holds:

    \textup{(i)} either $\Pi'$ has fewer memory states than $\Pi$, or $M'_H \triangle M'_L\subsetneq M_H \triangle M_L$

    \textup{(ii)} $U^R(\Pi)\leq U^R(\Pi')$.

\end{lemma}

\begin{proof}
    Given a simple profile $(\Pi,\sigma)$, where $\Pi=(f,g,a)$, we define a new action rule $a'$ as a solution of the following optimization problem:
    \begin{align}
    \sup_{\tilde{a}} \quad & U^R((f,g,\tilde{a}),\sigma) \notag \\
    \text{s.t.} \quad & \tilde{a}(i) = a(i) \text{ for all } i \notin M_H \triangle M_L, \label{eq:constraint1} \\
    & \tilde{a}(i) \in [a(1),a(m)] \text{ for all } i \in M_H \triangle M_L \label{eq:constraint2} \\
    & \sigma \in \br(f,g,\tilde{a}). \label{eq:constraint3}
\end{align}

To express constraint \eqref{eq:constraint3}, we write $U^R((f, g, \tilde{a}), \sigma)$ explicitly as a function of $\sigma$. Note that, given $(f, g)$ and $\sigma$, the probability of the game stopping in each memory state is independent of $\tilde{a}$. Therefore, $U^R((f, g, \tilde{a}), \sigma)$ is affine in $\tilde{a}$. The sender’s payoff $U^S((f, g, \tilde{a}), \sigma)$ is also affine in $\tilde{a}$, and thus the best response condition \eqref{eq:constraint3} imposes a linear constraint on $\tilde{a}$. It follows that the problem is a linear program. 

Since $a$ is feasible, a solution $a'$ exists. We may take $a'$ to be an extreme point of the feasible set defined by the constraints, and by definition, $$U^R(\Pi)=U^R((f, g, a), \sigma) \leq U^R((f, g, a'), \sigma).$$ Because $a'$ is an extreme point, one of the constraints \eqref{eq:constraint2} or \eqref{eq:constraint3} must bind. We now consider two cases:

      \medskip
    \noindent\underline{Case 1:} $a'$ is an extreme point of \eqref{eq:constraint3}.

    \smallskip
    In this case, there exists a memory state $\tilde{i} \in M_H \triangle M_L$ and a state of nature $\tilde{\theta}$ such that the sender is indifferent between stopping and continuing in $\tilde{i}$ under $\tilde{\theta}$. Define $\sigma'$ by:
    \begin{equation}
        \sigma'(i, \theta) = 
        \begin{cases}
            \sigma(i, \tilde{\theta}') & \text{if } i = \tilde{i} \text{ and } \theta =\tilde{\theta}, \\
            \sigma(i, \theta)         & \text{otherwise},
        \end{cases}
        \label{eq:sigma_prime_def}
    \end{equation} where $\tilde{\theta}'\neq \tilde{\theta}.$ That is, in state $\tilde{i}$, the sender now takes the same action in both states of nature. Then $\sigma'$ is a best response to $\Pi' = (f, g, a')$.

 If $\sigma'(\tilde{i}, \cdot) \equiv 0$, i.e., the sender continues in $\tilde{i}$ under both states of nature, then
 applying \Cref{lem:transient-state-action} yields a protocol $\Pi''$ and $M_H'' \triangle M_L'' \subsetneq M_H \triangle M_L$.

If $\sigma'(\tilde{i},\cdot)=1$, i.e., the sender stops in $\tilde{i}$ under both states of nature, then by \Cref{lem:not-stopping-at-transient}, we can eliminate one memory state from $\{1,m,\tilde{i}\}$ from $\Pi'$ to obtain $\Pi''$,  which has strictly fewer memory states than $\Pi$.

 Restricting $\sigma'$ to $\Pi''$ yields a strategy $\sigma''$. The resulting profile $(\Pi'', \sigma'')$ is simple and satisfies the desired properties.

\medskip
    \noindent\underline{Case 2:} $a'$ is an extreme point of \eqref{eq:constraint2}.

    \smallskip
    In this case, there exists a memory state $\hat{i}$ such that $a'(\hat{i}) = a(1)$ or $a'(\hat{i}) = a(m)$. By \Cref{def:simple-profile}, \Cref{lem:bound-of-stopping-states}, and constraint \eqref{eq:constraint3}, $a'(i) \in [a(1),a(m)]$ for all $i$. 

    If $a'(\hat{i}) = a(1)$, then there is a best response $\sigma'$ with $\sigma'(\hat{i}, \theta) = 0$ for both $\theta$, i.e., the sender continues in $\hat{i}$ in both states of nature. Applying \Cref{lem:transient-state-action} yields a protocol $\Pi''$ and $M_H'' \triangle M_L'' \subsetneq M_H \triangle M_L$.

    If $a'(\hat{i}) = a(m)$, then there is a best response $\sigma'$ with $\sigma'(\hat{i}, \theta) = 1$ for both $\theta$, i.e., the sender stops in $\hat{i}$ in both states of nature. Applying \Cref{lem:not-stopping-at-transient} yields a protocol $\Pi''$, which has strictly fewer memory states than $\Pi$.

   Restricting $\sigma'$ to $\Pi''$ yields a strategy $\sigma''$. The resulting profile $(\Pi'', \sigma'')$ is simple and satisfies the desired properties.
\end{proof}

The two lemmas above imply the following:

\begin{corollary}\label{cor:simple-profile}
    To study $\sup_{\Pi} U^R(\Pi)$, it is without loss of optimality to restrict attention to simple profiles $(\Pi, \sigma)$ defined on $\{1,\ldots,m\}$ such that $a_{\min} = a(1) = \dots = a(m-1) \leq a(m) = a_{\max}$ and $M_H = M_L = \{1, m\}$.
\end{corollary}

\begin{proof}[{Proof of \Cref{prop:simple-to-parsimonious} and \Cref{thm:parsimonious}}]

Given any simple profile $(\Pi, \sigma)$ satisfying \Cref{cor:simple-profile}, where $\Pi = (f, g, a)$, suppose $a_{\min} < a_{\max}$. Define a normalized action rule by
\[
a'(i) = \frac{a(i) - a_{\min}}{a_{\max} - a_{\min}}\in \{0,1\}.
\]
Then $\sigma$ is a best response to $\Pi' = (f, g, a')$, which is a parsimonious protocol. Therefore,
\[
\begin{aligned}
    U^R(\Pi) &= a_{\min} p + (1 - a_{\max})(1 - p) + (a_{\max} - a_{\min}) U^R(\Pi') \\
           &\leq \max\{p, 1 - p, U^R(\Pi')\}.
\end{aligned}
\]

Now consider a parsimonious protocol $\Pi^0$ in which the absorbing states are assigned $a(1) = 0$ and $a(m) = 1$. Let the initial state be $1$ if $p \leq \frac{1}{2}$ and $m$ otherwise. Then
\[
U^R(\Pi^0) = \max\{p, 1 - p\}.
\]
It follows that
\[
U^R(\Pi) \leq \max\{U^R(\Pi^0), U^R(\Pi')\}.
\]
If instead $a_{\min} = a_{\max}$, then
\[
U^R(\Pi) = p a_{\min} + (1 - p)(1 - a_{\min}) \leq \max\{p, 1 - p\} = U^R(\Pi^0).
\]
This proves \Cref{prop:simple-to-parsimonious}, and hence \Cref{thm:parsimonious}.
\end{proof}

\section{Proof of \Cref{thm:supUR}}\label{sec: proof with modified chain}

\subsection{Preview}\label{sec:thm2-overview}

As discussed in \Cref{sec:optimal value}, \Cref{thm:supUR} is proved using a relaxed problem. We first present the following property about a parsimonious protocol that defines the key constraint of the relaxed problem.
\begin{lemma}\label{lem:spread}
    For any parsimonious protocol $\Pi$, we have
    \begin{equation}\label{eqn:spread}
        \mu_m^{H}(\Pi)\mu_1^{L}(\Pi)\leq \gamma^{m-2}\mu_1^{H}(\Pi)\mu_m^{L}(\Pi),
    \end{equation} or equivalently, \[
{\mu_m^H(\Pi)}
{\mu_1^L(\Pi)}
\leq
\gamma^{m-2}\left({1-\mu_m^H(\Pi)}\right)\left({1-\mu_1^L(\Pi)}\right).
\]
Furthermore, if $m\geq 4$ and $\mu_i^\theta(\Pi)\neq 0$ for all $i\in\{1,m\}$ and $\theta=\{H,L\}$, then the strict inequality holds in \eqref{eqn:spread}. 
\end{lemma}

Using this new constraint in \eqref{eqn:spread}, we formulate a relaxed version of the receiver's optimization problem: \begin{equation} \tag{RR} \label{eqn:receiver_opt_relax}
    \begin{aligned}
        \max_{(\alpha,\beta)} \ \ & p\mspace{1mu}\alpha+(1-p)\mspace{1mu}\beta  \\
        \text{s.t.} \ \ \ & \alpha\beta  \leq \gamma^{m-2}(1-\alpha)(1-\beta),\\
        &  0\leq  \alpha, \beta \leq 1.
    \end{aligned}
\end{equation}
This optimization problem is straightforward to solve:

\begin{lemma}\label{lem:auxiliary-optimization}
The optimal value for the optimization problem \eqref{eqn:receiver_opt_relax} is

\begin{equation*}
U^{RR}(p,\gamma,m) =
\begin{cases}
\displaystyle
\bar p
+
\underline p\,
\frac{
\left(\gamma^{(m-2)/2}-\sqrt{\kappa}\right)^2
}{
\gamma^{m-2}-1
},
&
\textup{if }\gamma^{m-2}>\kappa,
\\[1.3em]
\bar p,
&
\textup{otherwise.}
\end{cases}
\end{equation*}
Furthermore, the first constraint in \eqref{eqn:receiver_opt_relax} must bind at an optimal solution $(\alpha^*,\beta^*)$; the optimal solution is unique and interior if $\gamma^{m-2} > \kappa$.
\end{lemma}

By \Cref{lem:spread}, for any $\Pi$, $(\mu_m^H(\Pi),\mu_1^L(\Pi))$ is feasible for \eqref{eqn:receiver_opt_relax} and hence its optimal value provides an upper bound of the receiver's optimal payoff. In the learning region, the relaxed problem has a unique
interior optimizer. For $m\geq4$, the strict inequality in
\Cref{lem:spread} rules out this optimizer as the
absorption-probability pair of any parsimonious protocol.
Hence no parsimonious protocol attains the relaxed value. The following result is therefore immediate from \Cref{lem:spread} and \Cref{lem:auxiliary-optimization}.

\begin{corollary}\label{lem:upper-bound-parsimonious}
    For any parsimonious protocol $\Pi$, $U^{R}(\Pi)\leq U^{RR}(p,\gamma,m),$ where the inequality is strict if $m\geq 4$ and $\gamma^{m-2}>\kappa$.
\end{corollary}

If $\gamma^{m-2} \leq \kappa$, the upper bound $\max\{p,1-p\}$ can be achieved by a trivial protocol. To show that the upper bound in \Cref{lem:upper-bound-parsimonious} is indeed the supremum when $\gamma^{m-2} > \kappa$, which necessitates $m > 2$, we construct a sequence of parsimonious protocols $\Pi(\epsilon_1, \epsilon_2)$ as follows. The initial state is given by $g(2) = 1$. The action rule is defined by $a(m) = 1$ and $a(i) = 0$ for all $i \neq m$. The transition rule is specified by:
\begin{equation*}
\begin{aligned}
f(i, l)(j) &=
\begin{cases}
    1, & \text{if } i = j \in \{1, m\}, \\
    {\left(\pi^H(h)\pi^L(h)\right)^{\frac{m-2}{2}}}\epsilon_1, & \text{if } i = 2,\, j = 1, \\
    1 - {\left(\pi^H(h)\pi^L(h)\right)^{\frac{m-2}{2}}}\epsilon_1, & \text{if } i = j = 2, \\
    1, & \text{if } i = j + 1 \notin \{2, m\}, \\
    0, & \text{otherwise},
\end{cases} \\[1em]
f(i, h)(j) &=
\begin{cases}
    1, & \text{if } i = j \in \{1, m\}, \\
    {\left(\pi^H(l)\pi^L(l)\right)^{\frac{m-2}{2}}} \epsilon_2, & \text{if } i = m - 1,\, j = m, \\
    1 - {\left(\pi^H(l)\pi^L(l)\right)^{\frac{m-2}{2}}} \epsilon_2, & \text{if } i = j = m - 1, \\
    1, & \text{if } i = j - 1 \notin \{1, m - 1\}, \\
    0, & \text{otherwise}.
\end{cases}
\end{aligned}
\end{equation*}
For every $s\notin\{h,l\}$, set $f(i,s)(i)=1$ for all $i\in M$. This parsimonious protocol features a random transition in memory states $2$ and $m-1$. By choosing $(\epsilon_1, \epsilon_2)$ appropriately, the following result completes the proof of \Cref{thm:supUR}.
\begin{lemma}\label{lem:setting_k2/k1}
Assume $\gamma^{m-2}>\kappa$. Then there exists $k_1,k_2>0$ such that
\[
\lim_{\epsilon \to 0}U^R(\Pi(k_1\epsilon, k_2\epsilon))=
   1-\frac{2\sqrt{p(1-p)\gamma^{m-2}}-1}{\gamma^{m-2}-1}=\bar p
+
\underline p\,
\frac{
\left(\gamma^{(m-2)/2}-\sqrt{\kappa}\right)^2
}{
\gamma^{m-2}-1
}.
\]
\end{lemma}

Together, \Cref{lem:auxiliary-optimization}, \Cref{lem:setting_k2/k1}, and \Cref{lem:upper-bound-parsimonious},
combined with the optimality of parsimonious protocols established
in \Cref{thm:parsimonious}, establish the receiver's optimal payoff
and the non-attainment result, proving parts \textup{(i)} and
\textup{(ii)} of \Cref{thm:supUR}.

The remaining subsections prove these supporting results.
We first introduce the modified Markov chain technique of
\citet{hellman1969learning}.

\subsection{Preliminaries: Modified Markov Chain}\label{sec:prel-modified-markov}
For a finite-state Markov chain $\mathcal{M}$ with transition probabilities $p_{ij}$, let $T(\mathcal{M})$ denote the set of its transient states and fix a state $m_0 \in T(\mathcal{M})$. Let $R$ be a recurrent communicating class of $\mathcal{M}$ and let $P_R$ denote the probability that $\mathcal{M}$ is eventually absorbed into $R$ starting from the initial state $m_0$. Let $\tau$ be the random time the process escapes the transient set starting from $m_0$.

To compute $P_R$, define a \textbf{modified Markov chain} obtained by restricting $\mathcal{M}$  to $T(\mathcal{M})$, where the process starts at $m_0$, and any transition from a transient state to a recurrent state in the original chain is redirected instead to $m_0$. Let $\nu$ denote the stationary distribution of this modified Markov chain on $T(\mathcal{M})$. \citet{hellman1969learning} (from A4 to A8, pp.48-49) proved the following:
\begin{equation}\label{lem:absorbing-formula}
P_R= \mathbb{E}[\tau] \sum_{i \in T(\mathcal{M})} \nu_i \sum_{j \in R} p_{ij}.
\end{equation}

The same formula holds for an initial distribution $g$
supported on the transient set when every exit is redirected
to an independent draw from $g$. 
The following is due to \citet{hellman1970learning} (Theorem 2 and Corollary 1).

\begin{lemma} \label{lem:condition-for-irreducible} 

For any protocol with $m-2$ states that induces an irreducible Markov chain with stationary distribution $\nu^\theta$ when the state is $\theta \in \{H,L\}$, the following holds for any two memory states $i,j$,

    \begin{equation}\label{hellman-spread-bound}
        \frac{\nu_i^H \nu_j^L}{\nu_i^L\nu_j^H}\leq \gamma^{m-3}.
    \end{equation}

   Label the $m-2$ states as $2,3,...,m-1$ ranked increasingly by the ratios ${\nu^H_i}/{\nu^L_i}$, the following must be satisfied for the equality to hold in \eqref{hellman-spread-bound}:

    \textup{(i)} Any transition from $i$ to $i+1$ is due to signal $h$ with likelihood ratio $\bar{\ell}$.

    \textup{(ii)}  Any transition from $i$ to $i-1$ is due to signal $l$ with likelihood ratio $\underline{\ell}$.

    \textup{(iii)}  No transitions from $i$ to $j$ exist if $|i-j|\geq 2$.
\end{lemma}

\subsection{Proof of \Cref{lem:spread}}

\begin{proof}
     We first consider deterministic initialization, with initial memory state being ${i_0}$.  First note that $\mu_i^H(\Pi)=0$ if and only if $\mu_i^L(\Pi)=0$ for $i=1,m.$ Therefore, the claim is true if one of the absorption probabilities is zero. Now suppose $\mu_i^\theta(\Pi)\neq 0$ for all $i \in\{1,m\}$ and $\theta\in \{H,L\}$. We may assume that every transient state is reachable from
$i_0$: otherwise, delete unreachable states and apply the
calculation below to the resulting $m'<m$-state protocol,
obtaining the stronger bound $\gamma^{m'-2}<\gamma^{m-2}$
without changing the absorption probabilities.
The modified chain is then irreducible, so
$\nu_i^\theta>0$ for every transient state.
In extrema of transition likelihood ratios, omit transitions
with zero denominator; full support implies that their
numerators are also zero. Let the stationary distribution of $T(\Pi)$ when the state of nature is $\theta$ be $\nu^\theta$. Use $\tau^\theta$ to denote the time to absorption when the state of nature is $\theta$. Then
    \begin{align*}
        \frac{\mu_m^{H}(\Pi)\mu_1^{L}(\Pi)}{\mu_1^{H}(\Pi)\mu_m^{L}(\Pi)}&=\frac{\mathbb{E}[\tau^H](\sum_{2\leq i\leq m-1} \nu^H_i p^H_{im})\mathbb{E}[\tau^L](\sum_{2\leq i\leq m-1} \nu^L_i p^L_{i1})}{\mathbb{E}[\tau^L](\sum_{2\leq i\leq m-1} \nu^L_i p^L_{im})\mathbb{E}[\tau^H](\sum_{2\leq i\leq m-1} \nu^H_i p^H_{i1})}\\
        &\leq \frac{\max_{2\leq i\leq m-1} \frac{p^H_{im}}{p^L_{im}}}{\min_{2\leq j\leq m-1} \frac{p^H_{j1}}{p^L_{j1}}}\cdot \frac{\max_{2\leq i\leq m-1} \frac{\nu^H_i}{\nu^L_i}}{\min_{2\leq j\leq m-1}\frac{\nu^H_j}{\nu^L_j}}\\
        &\leq \frac{\bar{\ell}}{\underline{\ell}}\gamma^{m-3}\\
        &=\gamma^{m-2}
    \end{align*}    
  For the first inequality to hold with equality, the transition from a transient state $i$ to $m$ can only happen if $i$ maximizes $\frac{\nu^H_i}{\nu^L_i}$ and the transition from a transient state $j$ to $1$ can only happen if $j$ minimizes $\frac{\nu^H_j}{\nu^L_j}$. For the second inequality to hold with equality, the transition from a transient state to $m$ is due to the signal $h$ with likelihood $\bar{\ell}$, the transition from a transient state to $1$ is due to the signal $l$ with likelihood $\underline{\ell}$, and \eqref{hellman-spread-bound} holds with equality for $T(\Pi)$ with $m-2$ memory states. Label the transient states as $2,3,...,m-1$, ranked by the ratios $\frac{\nu^H_i}{\nu^L_i}$. Then for $\Pi$, there is a transition due to $l$ from $2$ to $1$, and a transition from $m-1$ to $m$ due to $h$. So for $T(\Pi)$, there is a transition due to $l$ from $2$ to $i_0$ and a transition due to $h$ from $m-1$ to $i_0$. For $m\geq 4$, it cannot be the case that $2, i_0, m-1$ are the same state. It follows from by \Cref{lem:condition-for-irreducible} that \eqref{hellman-spread-bound}, and hence the second inequality here, cannot hold with equality.

    It remains to consider an arbitrary initial distribution $g$. For arbitrary initialization, the pair
$(\mu_m^H(\Pi),\mu_1^L(\Pi))$ is a convex combination of the
pairs obtained under deterministic initializations. The region
defined by \eqref{eqn:spread} is convex, since its upper boundary
\[
\beta=
\frac{\gamma^{m-2}(1-\alpha)}
{\gamma^{m-2}-(\gamma^{m-2}-1)\alpha}
\]
is concave. Hence the weak bound extends to arbitrary initial
distributions. For $m\geq4$, this boundary is strictly concave,
so equality at a pair in $(0,1)^2$ would require every positively
weighted component to be that same boundary point, contradicting
the strict bound for deterministic initialization. This establishes the strict bound as well.

\end{proof}

\subsection{Proof of \Cref{lem:auxiliary-optimization}}\label{proof:lem:auxiliary-optimization}

\begin{proof}

By inspection, in the optimal solution, the first inequality constraint in \eqref{eqn:receiver_opt_relax} must be binding; otherwise, it must be that \(\alpha, \beta < 1\), and both variables can be increased to raise the objective value. 

We shall prove that the following four statements hold for the optimization problem \eqref{eqn:receiver_opt_relax}, which imply \Cref{lem:auxiliary-optimization}: 

\textup{(i)} If $\gamma^{m-2} > \kappa$, then the unique solution is
\begin{equation} \label{eq:alphabeta}
\alpha^* = \frac{\gamma^{m-2} - \sqrt{\frac{1 - p}{p} \gamma^{m-2}}}{\gamma^{m-2} - 1}, \quad
\beta^* = \frac{\gamma^{m-2} - \sqrt{\frac{p}{1 - p} \gamma^{m-2}}}{\gamma^{m-2} - 1}.
\end{equation}
Both are interior and the optimal value is
\begin{equation}\label{eq:optimization_value}
1 - \frac{2\sqrt{p(1 - p)\gamma^{m-2}} - 1}{\gamma^{m-2} - 1},
\end{equation}
which is strictly greater than $\max\{p, 1 - p\}$ and converges to it as $\gamma^{m-2} \to \kappa$.

\textup{(ii)} If $\gamma^{m-2} \leq \kappa$ and $p > \frac{1}{2}$, then the unique solution is
\begin{equation}
(\alpha^*, \beta^*) = (1, 0),
\end{equation}
and the optimal value is \( p \).

\textup{(iii)} If $\gamma^{m-2} \leq \kappa$ and $p < \frac{1}{2}$, then the unique solution is
\begin{equation}
(\alpha^*, \beta^*) = (0, 1),
\end{equation}
and the optimal value is \( 1 - p \).

\textup{(iv)} If $\gamma^{m-2} \leq \kappa$ and $p = \frac{1}{2}$, then the optimal value is \( \frac{1}{2} \), and the solution is not unique.

Solving the optimization problem when the first constraint binds yields \((\alpha^*, \beta^*)\) of the form given in \eqref{eq:alphabeta}, and we have \(\alpha^*, \beta^* \in (0,1)\) if and only if \(\gamma^{m-2} > \kappa\). In this case, the optimal value \(p\alpha^* + (1 - p)\beta^*\) simplifies to the expression in \eqref{eq:optimization_value}. A straightforward computation shows that this value is strictly greater than \(\max\{p, 1 - p\}\), and the two coincide if \(\gamma^{m-2} = \kappa\).

Now suppose $\gamma^{m-2}\leq\kappa$. If $p>\tfrac12$, then
$\kappa=p/(1-p)$. The feasibility constraint implies
\[
\beta
\leq
\frac{\gamma^{m-2}(1-\alpha)}
{1+(\gamma^{m-2}-1)(1-\alpha)}
\leq
\gamma^{m-2}(1-\alpha).
\]
Consequently,
\[
\begin{aligned}
p\alpha+(1-p)\beta
&\leq p\alpha+(1-p)\gamma^{m-2}(1-\alpha)\\
&\leq p\alpha+p(1-\alpha)\\
&=p.
\end{aligned}
\]
If $\alpha<1$, at least one of these inequalities is strict.
When $\gamma^{m-2}<\kappa$, the second inequality in the
payoff calculation is strict. When $\gamma^{m-2}=\kappa>1$,
the bound $\beta\leq\gamma^{m-2}(1-\alpha)$ is strict.
Thus the unique optimizer is $(\alpha^*,\beta^*)=(1,0)$,
with value $p$.

The case $p<\tfrac12$ follows symmetrically by interchanging
$\alpha$ with $\beta$ and $p$ with $1-p$. The unique optimizer
is $(\alpha^*,\beta^*)=(0,1)$, with value $1-p$.

Finally, if $p=\tfrac12$, then $\kappa=1$, so
$\gamma^{m-2}\leq\kappa$ implies $\gamma^{m-2}=1$.
The feasibility constraint then reduces to $\alpha+\beta\leq 1.$
So the optimal value is
$\tfrac12$, attained by every pair
$(\alpha,\beta)\in[0,1]^2$ satisfying $\alpha+\beta=1$.
This completes the proof.
\end{proof}

\subsection{Proof of \Cref{lem:setting_k2/k1}}

\begin{lemma}\label{lem:upper-bound}
    For any $k_1,k_2>0$, as $\epsilon\rightarrow 0$, we have

    \[\mu_m^H(\Pi(k_1\epsilon, k_2\epsilon))\rightarrow \frac{\frac{k_2}{k_1}\gamma^{\frac{m-2}{2}}}{1+\frac{k_2}{k_1}\gamma^{\frac{m-2}{2}}} \quad \text{and} \quad  \mu_1^L(\Pi(k_1\epsilon, k_2\epsilon))\rightarrow \frac{\frac{k_1}{k_2}\gamma^{\frac{m-2}{2}}}{1+\frac{k_1}{k_2}\gamma^{\frac{m-2}{2}}}.\]

\end{lemma}

\begin{proof}
  For each state of nature $\theta$, the protocol $\Pi(k_1\epsilon,k_2\epsilon)$ defines a Markov chain on $M$ with the initiate state $m_0=2$. Let $\nu^\theta(\epsilon)$ be the stationary distribution of the modified Markov chain obtained by restricting $\Pi(k_1\epsilon,k_2\epsilon)$ to $T(\Pi(k_1\epsilon,k_2\epsilon))=\{2,...,m-1\}$ when the state of nature is $\theta$. By the definition of $\Pi(k_1\epsilon,k_2\epsilon)$, the modified Markov chain on $\{2,...,m-1\}$ is irreducible for any $\epsilon\geq 0$. Therefore, $\nu^\theta(\epsilon)\rightarrow\nu^\theta(0)$ if $\epsilon\rightarrow 0$. Considering the balance equations on the modified Markov chain,  $\pi^\theta(l)\nu_{i+1}^\theta(0)=\pi^\theta(h)\nu_i^\theta(0)$ for $i=2,...,m-2$, we have $\nu_{i+1}^\theta(0)=\frac{\pi^\theta(h)}{\pi^\theta(l)}\nu_i^\theta(0)$. Therefore,
  \begin{equation}\label{eqn:balance}
  \nu^\theta_{m-1}(0)=\left(\frac{\pi^\theta(h)}{\pi^\theta(l)}\right)^{m-3}\nu^\theta_2(0).
  \end{equation}
Let $\tau^\theta(\epsilon)$ be the random time the Markov process $\Pi(k_1\epsilon,k_2\epsilon)$ escapes $\{2,...,m-1\}$ when the state of nature is $\theta$. Now applying \Cref{lem:absorbing-formula} to $R=\{1\}$ and $R=\{m\}$, respectively, we obtain
  \begin{equation}\label{eqn:steady-state-k}
\begin{aligned}
\mu_1^\theta(\Pi(k_1\epsilon, k_2\epsilon)) 
&= \mathbb{E}[\tau^\theta(\epsilon)]  \nu^\theta_{2}(\epsilon)   \pi^\theta(l) \left(\pi^H(h)\pi^L(h)\right)^{\frac{m-2}{2}} k_1 \epsilon, \\
\mu_m^\theta(\Pi(k_1\epsilon, k_2\epsilon)) 
&= \mathbb{E}[\tau^\theta(\epsilon)]   \nu^\theta_{m-1}(\epsilon)   \pi^\theta(h) \left(\pi^H(l)\pi^L(l)\right)^{\frac{m-2}{2}} k_2 \epsilon.
\end{aligned}
\end{equation}
    Since $\mu_1^\theta(\Pi(k_1\epsilon, k_2\epsilon))+\mu_m^\theta(\Pi(k_1\epsilon, k_2\epsilon))=1$, for $i\neq j\in\{1,m\}$, we have $$\mu_i^\theta(\Pi(k_1\epsilon, k_2\epsilon))=\frac{\mu_i^\theta(\Pi(k_1\epsilon, k_2\epsilon))}{\mu_i^\theta(\Pi(k_1\epsilon, k_2\epsilon))+\mu_j^\theta(\Pi(k_1\epsilon, k_2\epsilon))}.$$ 
Using \eqref{eqn:balance} and then \eqref{eqn:steady-state-k}, we have
\begin{align*}
\mu_m^H(\Pi(k_1\epsilon, k_2\epsilon)) 
&= \frac{ \nu^H_{m-1}(\epsilon) \, \pi^H(h) \left(\pi^H(l)\pi^L(l)\right)^{\frac{m-2}{2}} k_2 }
         { \nu^H_{m-1}(\epsilon) \, \pi^H(h) \left(\pi^H(l)\pi^L(l)\right)^{\frac{m-2}{2}} k_2 
         + \nu^H_{2}(\epsilon) \, \pi^H(l) \left(\pi^H(h)\pi^L(h)\right)^{\frac{m-2}{2}} k_1 } \\[0.5em]
&= \frac{ \frac{k_2}{k_1} \cdot \frac{\nu^H_{m-1}(\epsilon)}{\nu^H_2(\epsilon)} 
\left( \frac{\pi^H(l)\pi^L(l)}{\pi^H(h)\pi^L(h)} \right)^{\frac{m-3}{2}} 
\left({ \frac{ \pi^H(h)\pi^L(l) }{ \pi^H(l)\pi^L(h) } }\right)^{\frac{1}{2}} }
{ \frac{k_2}{k_1} \cdot \frac{\nu^H_{m-1}(\epsilon)}{\nu^H_2(\epsilon)} 
\left( \frac{\pi^H(l)\pi^L(l)}{\pi^H(h)\pi^L(h)} \right)^{\frac{m-3}{2}} 
\left({ \frac{ \pi^H(h)\pi^L(l) }{ \pi^H(l)\pi^L(h) }} \right)^{\frac{1}{2}} + 1 } \\[0.5em]
&\rightarrow \frac{ \frac{k_2}{k_1} \left( \frac{\pi^H(h)}{\pi^H(l)} \right)^{m-3}
\left( \frac{\pi^H(l)\pi^L(l)}{\pi^H(h)\pi^L(h)} \right)^{\frac{m-3}{2}} 
\gamma^{\frac{1}{2}} }
{ \frac{k_2}{k_1} \left( \frac{\pi^H(h)}{\pi^H(l)} \right)^{m-3}
\left( \frac{\pi^H(l)\pi^L(l)}{\pi^H(h)\pi^L(h)} \right)^{\frac{m-3}{2}} 
\gamma^{\frac{1}{2}} + 1 } \\[0.5em]
&= \frac{ \frac{k_2}{k_1} \gamma^{\frac{m-2}{2}} }
{ \frac{k_2}{k_1} \gamma^{\frac{m-2}{2}} + 1 }.
\end{align*}
    The limit for $\mu_1^L(\Pi(k_1\epsilon, k_2\epsilon))$ follows symmetrically.
\end{proof}

\begin{proof}[Proof of \Cref{lem:setting_k2/k1}]
By setting
\[
\frac{k_2}{k_1}
   =\frac{\sqrt{\gamma^{m-2}}-\sqrt{\tfrac{1-p}{p}}}%
          {\sqrt{\gamma^{m-2}}\sqrt{\tfrac{1-p}{p}}-1}
\] in \Cref{lem:upper-bound}, we have $\mu_m^H(\Pi(k_1\epsilon, k_2\epsilon)) \rightarrow \alpha^* $ and $\mu_1^L(\Pi(k_1\epsilon, k_2\epsilon)) \rightarrow \beta^*$, 
and hence $$U^R(\Pi(k_1\epsilon, k_2\epsilon))\rightarrow p\alpha^*+(1-p)\beta^*=1-\frac{2\sqrt{p(1-p)\gamma^{m-2}}-1}{\gamma^{m-2}-1}.$$ This completes the proof.
\end{proof}

\section{Proof of \Cref{cor:interior-absorption}}

\begin{proof}
By \Cref{lem:spread}, each pair
$(\mu_m^H(\Pi^n),\mu_1^L(\Pi^n))$ belongs to the compact
feasible set of the relaxed problem
\eqref{eqn:receiver_opt_relax}. By asymptotic optimality and
the value characterization in \Cref{thm:supUR},
\[
p\mu_m^H(\Pi^n)+(1-p)\mu_1^L(\Pi^n)
=
U^R(\Pi^n)
\rightarrow
U^{RR}(p,\gamma,m).
\]
Continuity of the objective implies that every accumulation
point of the absorption-probability pairs solves the relaxed
problem. When $\gamma^{m-2}>\kappa$, that problem has the unique
optimizer $(\alpha^*,\beta^*)\in(0,1)^2$ given by
\eqref{eq:alphabeta}, as established in
\Cref{lem:auxiliary-optimization}. Hence
\[
\bigl(\mu_m^H(\Pi^n),\mu_1^L(\Pi^n)\bigr)
\rightarrow(\alpha^*,\beta^*).
\]

Since
$\mu_1^\theta(\Pi^n)+\mu_m^\theta(\Pi^n)=1$
for each $\theta$, we also have
\[
\mu_1^H(\Pi^n)\rightarrow1-\alpha^*,
\qquad
\mu_m^L(\Pi^n)\rightarrow1-\beta^*.
\]
All four limits therefore belong to $(0,1)$.
\end{proof}

\section{Proof of \Cref{thm:behavioral_implication}}\label{appendix: proof-behavior}

\textbf{Overview.}
We first characterize asymptotically optimal sequences of parsimonious
protocols in \Cref{lem:3-conditions}. The characterization separates
three requirements: optimal absorption, optimal mixing, and optimal
bias. Optimal absorption specifies the states and signals through
which absorption occurs. Optimal mixing requires the modified chain
on transient states to approach its likelihood-ratio spread bound.
Together, these conditions make the absorption likelihood-ratio
spread approach $\gamma^{m-2}$, while optimal bias ensures the
balance between terminal decisions required by the prior.

The optimal absorption condition yields part \textup{(i)} of
\Cref{thm:behavioral_implication}. To establish part \textup{(ii)},
\Cref{lem:part_ii_thm} shows that any absorption transition remaining
bounded away from zero violates at least one of the three conditions.
We combine these results in the proof at the end of this section.

In this section, let us consider a sequence of $m$-state parsimonious protocols
$\{\Pi^n=(f^n,g^n,a^n)\}_{n=1}^\infty$. Write $T:=\{2,\ldots,m-1\}$.
Whenever $g^n(T)>0$, define
\[
\widehat g^n(j):=\frac{g^n(j)}{g^n(T)},
\qquad j\in T,
\]
and the modified transition rule
\[
\widetilde f^n(i,s)(j)
:=
f^n(i,s)(j)
+
f^n(i,s)(\{1,m\})\widehat g^n(j),
\qquad i,j\in T.
\]
Let $\nu^\theta(\Pi^n)$ denote the stationary distribution of
this modified chain under $\theta$, and order the transient
states by $\nu_i^H(\Pi^n)/\nu_i^L(\Pi^n)$.

\begin{proposition}\label{lem:3-conditions}
    Suppose $\gamma^{m-2}>\kappa$ and $\{\Pi^n\}_{n=1}^\infty$ is a sequence of parsimonious protocols. Then $\lim\limits_{n \rightarrow \infty} U^R(\Pi^n) = \sup_\Pi U^R(\Pi)$ if and only if the following hold. All ratios below are required to be
well defined for all sufficiently large $n$.

    \textup{(i)} \textbf{Optimal Absorption}:

    \begin{equation*}
        \frac{\nu^\theta_2(\Pi^n)\pi^\theta(l)f^n(2,l)(1)}{\sum_{i=2}^{m-1} \sum_{s\in S}\nu^\theta_i(\Pi^n)\pi^\theta(s)f^n(i,s)(1)}\rightarrow 1;
    \end{equation*}

    \begin{equation*}
        \frac{\nu^\theta_{m-1}(\Pi^n)\pi^\theta(h)f^n(m-1,h)(m)}{\sum_{i=2}^{m-1} \sum_{s\in S}\nu^\theta_i(\Pi^n)\pi^\theta(s)f^n(i,s)(m)}\rightarrow 1.
    \end{equation*}

    \textup{(ii)} \textbf{Optimal Mixing}: 

    \textup{(ii.a)}
    \begin{equation*}
        \frac{\sum_{j=k+1}^{m-1}\sum_{s\in S}\nu^\theta_k(\Pi^n)\pi^\theta(s)\tilde{f}^n(k,s)(j)}{\sum_{i=2}^{k}\sum_{j=k+1}^{m-1}\sum_{s\in S}\nu^\theta_i(\Pi^n)\pi^\theta(s)\tilde{f}^n(i,s)(j)}\rightarrow 1;
    \end{equation*}
    \begin{equation*}
        \frac{\sum_{j=2}^{k}\sum_{s\in S}\nu^\theta_{k+1}(\Pi^n)\pi^\theta(s)\tilde{f}^n(k+1,s)(j)}{\sum_{i=k+1}^{m-1}\sum_{j=2}^{k}\sum_{s\in S}\nu^\theta_i(\Pi^n)\pi^\theta(s)\tilde{f}^n(i,s)(j)}\rightarrow 1.
    \end{equation*}

   \textup{(ii.b)}
    \begin{equation*}
        \frac{\sum_{j=k+1}^{m-1}\sum_{s\in S} \pi^H(s) \tilde{f}^n(k,s)(j)}{\sum_{j=k+1}^{m-1}\sum_{s\in S} \pi^L(s)\tilde{f}^n(k,s)(j)}\rightarrow \bar{\ell};
    \end{equation*}

    \begin{equation*}
        \frac{\sum_{j=2}^{k}\sum_{s\in S} \pi^H(s)\tilde{f}^n(k+1,s)(j)}{\sum_{j=2}^{k}\sum_{s\in S} \pi^L(s)\tilde{f}^n(k+1,s)(j)}\rightarrow \underline{\ell}.
    \end{equation*}

    \textup{(iii)} \textbf{Optimal Bias}:
$g^n(1)+g^n(m)\rightarrow0$ and $\mu_m^H(\Pi^n)\rightarrow\alpha^*,
$ where $\alpha^*$ is given by \Cref{cor:interior-absorption}.

\end{proposition}

\begin{remark}
Let $\widehat\Pi^n:=(f^n,\widehat g^n,a^n)$, with
$\widehat g^n$ extended by zero on $\{1,m\}$.
The optimal absorption and optimal mixing conditions together
are equivalent to
\[
\frac{
\mu_m^H(\widehat\Pi^n)\mu_1^L(\widehat\Pi^n)
}{
\mu_m^L(\widehat\Pi^n)\mu_1^H(\widehat\Pi^n)
}
\rightarrow\gamma^{m-2}.
\]
Thus, these conditions characterize the maximal limiting
absorption spread conditional on initialization in the
transient set. The optimal bias condition rules out
nonvanishing initial absorption and selects the limiting
balance between terminal decisions required by the prior.
\end{remark}

\begin{remark}\label{rem:equal-likelihood-transient-states}
    By Theorem 6 of \citet{hellman1970learning}, the optimal mixing condition is equivalent to 
    \begin{equation*}
        \frac{\nu_{m-1}^H(\Pi^n) \nu_{2}^L(\Pi^n)}{\nu_{m-1}^L(\Pi^n)\nu_{2}^H(\Pi^n)}\rightarrow\gamma^{m-3}.
    \end{equation*}
    If this holds, then
    \begin{equation*}
        \frac{\nu_{i+1}^H(\Pi^n) \nu_{i}^L(\Pi^n)}{\nu_{i+1}^L(\Pi^n)\nu_{i}^H(\Pi^n)}\rightarrow\gamma.
    \end{equation*}
    for every $2\leq i\leq m-2$. This implies that there can only be finite number of terms in the sequence where two transient states have equal likelihood ratios $\frac{\nu^H_i}{\nu^L_i}$.
\end{remark}

\begin{proof}[Proof of \Cref{lem:3-conditions}]
Write $V_m:=\sup_\Pi U^R(\Pi)$ and
$g_{\mathrm{abs}}^n:=g^n(1)+g^n(m)$. Since
\[
U^R(\Pi^n)
\leq
g_{\mathrm{abs}}^n\bar p
+(1-g_{\mathrm{abs}}^n)V_m
\]
and $V_m>\bar p$, asymptotic optimality implies
$g_{\mathrm{abs}}^n\to0$. This convergence is also required
by condition \textup{(iii)}. Thus, in either direction of the
proof, $g^n(T)>0$ for all sufficiently large $n$.

Let $\widehat\Pi^n:=(f^n,\widehat g^n,a^n)$, extending
$\widehat g^n$ by zero on $\{1,m\}$. We have
\[
U^R(\Pi^n)
=
p g^n(m)+(1-p)g^n(1)
+(1-g_{\mathrm{abs}}^n)U^R(\widehat\Pi^n),
\]
and
\begin{equation}\label{eq:initial-absorption-decomposition}
\mu_j^\theta(\Pi^n)
=
g^n(j)
+(1-g_{\mathrm{abs}}^n)\mu_j^\theta(\widehat\Pi^n),
\qquad j\in\{1,m\}.
\end{equation}
Because $g_{\mathrm{abs}}^n\to0$, replacing $\Pi^n$ by
$\widehat\Pi^n$ changes payoffs and absorption probabilities
by quantities tending to zero. It leaves the modified chain,
and hence conditions \textup{(i)} and \textup{(ii)}, unchanged.
We may therefore suppress the hats and assume $g^n(T)=1$
throughout the remaining argument.
For the only-if direction, every transient state is reachable
from the initial distribution for all sufficiently large $n$.
Otherwise, along a subsequence, an unreachable state could be
deleted without changing the receiver's payoff, giving
\[
U^R(\Pi^n)\leq V_{m-1}<V_m.
\]
Here $V_{m-1}$ denotes the receiver's optimal value with at
most $m-1$ memory states, and the strict inequality follows
from \Cref{thm:supUR} in the learning region. This contradicts
asymptotic optimality.

For the if direction, conditions \textup{(i)} and
\textup{(ii.a)} likewise imply eventual reachability: their
positive limiting flow shares require positive stationary
mass at every transient state. Full support of the signal
distributions makes reachability the same under both states
of nature. Once every transient state is reachable, the modified
chain is irreducible, since each state can reach an absorption
transition and reset to the initial distribution. Its stationary
distribution therefore assigns strictly positive probability
to every transient state under either state of nature.
    By \Cref{lem:absorbing-formula},
 \begin{equation}\label{eqn:absorbing_prob}
    \begin{aligned}
        \frac{\mu_m^{H}(\Pi^n)\mu_1^{L}(\Pi^n)}{\mu_1^{H}(\Pi^n)\mu_m^{L}(\Pi^n)}&=\frac{\mathbb{E}[\tau^H(\Pi^n)](\sum_{2\leq i\leq m-1} \nu^H_i p^H_{im})\mathbb{E}[\tau^L(\Pi^n)](\sum_{2\leq i\leq m-1} \nu^L_i p^L_{i1})}{\mathbb{E}[\tau^L(\Pi^n)](\sum_{2\leq i\leq m-1} \nu^L_i p^L_{im})\mathbb{E}[\tau^H(\Pi^n)](\sum_{2\leq i\leq m-1} \nu^H_i p^H_{i1})}\\
        &=\frac{\sum_{i=2}^{m-1} \sum_{s\in S}\nu^H_i(\Pi^n)\pi^H(s)f^n(i,s)(m)}{\sum_{i=2}^{m-1} \sum_{s\in S}\nu^L_i(\Pi^n)\pi^L(s)f^n(i,s)(m)}\frac{\sum_{i=2}^{m-1} \sum_{s\in S}\nu^L_i(\Pi^n)\pi^L(s)f^n(i,s)(1)}{\sum_{i=2}^{m-1} \sum_{s\in S}\nu^H_i(\Pi^n)\pi^H(s)f^n(i,s)(1)}.
    \end{aligned}
    \end{equation}
    Now we are ready to prove the if direction. If (i) and (ii) hold, then
    \allowdisplaybreaks
    \begin{align*}
        \frac{\mu_m^{H}(\Pi^n)\mu_1^{L}(\Pi^n)}{\mu_1^{H}(\Pi^n)\mu_m^{L}(\Pi^n)}
        &=\frac{\sum_{i=2}^{m-1} \sum_{s\in S}\nu^H_i(\Pi^n)\pi^H(s)f^n(i,s)(m)}{\sum_{i=2}^{m-1} \sum_{s\in S}\nu^L_i(\Pi^n)\pi^L(s)f^n(i,s)(m)}\frac{\sum_{i=2}^{m-1} \sum_{s\in S}\nu^L_i(\Pi^n)\pi^L(s)f^n(i,s)(1)}{\sum_{i=2}^{m-1} \sum_{s\in S}\nu^H_i(\Pi^n)\pi^H(s)f^n(i,s)(1)}\\
        &=\frac{\pi^H(h)\pi^L(l)}{\pi^L(h)\pi^H(l)} \frac{\nu_{m-1}^H(\Pi^n) \nu_{2}^L(\Pi^n)}{\nu_{m-1}^L(\Pi^n)\nu_{2}^H(\Pi^n)}\frac{\nu^L_{m-1}(\Pi^n)\pi^L(h)f^n(m-1,h)(m)}{\sum_{i=2}^{m-1} \sum_{s\in S}\nu^L_i(\Pi^n)\pi^L(s)f^n(i,s)(m)}\\
        &\frac{\sum_{i=2}^{m-1} \sum_{s\in S}\nu^H_i(\Pi^n)\pi^H(s)f^n(i,s)(m)} {\nu^H_{m-1}(\Pi^n)\pi^H(h)f^n(m-1,h)(m)}\frac{\nu^H_2(\Pi^n)\pi^H(l)f^n(2,l)(1)}{\sum_{i=2}^{m-1} \sum_{s\in S}\nu^H_i(\Pi^n)\pi^H(s)f^n(i,s)(1)}\\
        &\frac{\sum_{i=2}^{m-1} \sum_{s\in S}\nu^L_i(\Pi^n)\pi^L(s)f^n(i,s)(1)} {\nu^L_2(\Pi^n)\pi^L(l)f^n(2,l)(1)}\\
        &\rightarrow \gamma \gamma^{m-3}=\gamma^{m-2}.
    \end{align*}
    This together with (iii) gives $\mu^L_1(\Pi^n)\rightarrow \beta^*$ and hence $U^R(\Pi^n)\rightarrow \sup_{\Pi} U^R(\Pi)$.

For the only-if direction, condition \textup{(iii)} follows
from the initial-distribution reduction and
\Cref{cor:interior-absorption}. For brevity, write
\[
L_i^n:=\frac{\nu_i^H(\Pi^n)}{\nu_i^L(\Pi^n)},
\qquad
\ell(s):=\frac{\pi^H(s)}{\pi^L(s)},
\qquad
R_n:=
\frac{\mu_m^H(\Pi^n)\mu_1^L(\Pi^n)}
     {\mu_1^H(\Pi^n)\mu_m^L(\Pi^n)}.
\]
By \Cref{cor:interior-absorption} and the binding constraint
at the relaxed optimizer, $R_n\rightarrow\gamma^{m-2}.$
On the other hand, \eqref{eqn:absorbing_prob} and
\eqref{hellman-spread-bound} imply
\[
R_n
\leq
\gamma\frac{L_{m-1}^n}{L_2^n}
\leq\gamma^{m-2}.
\]
Hence $L_{m-1}^n/L_2^n\to\gamma^{m-3}$, which gives
condition \textup{(ii)} by
\Cref{rem:equal-likelihood-transient-states}.

It remains to establish condition \textup{(i)}.
Under our convention that $h$ and $l$ are the unique
extreme-likelihood-ratio signals, we can choose
\[
\xi\in
\left(
\max\left\{
\frac1\gamma,\,
\max_{s\neq h}\frac{\ell(s)}{\bar\ell},\,
\max_{s\neq l}\frac{\underline\ell}{\ell(s)}
\right\},
1
\right).
\]
If $m\geq4$, condition \textup{(ii)} implies
\[
\frac{L_{m-2}^n}{L_{m-1}^n}\rightarrow\frac1\gamma,
\qquad
\frac{L_2^n}{L_3^n}\rightarrow\frac1\gamma.
\]
Both ratios are therefore at most $\xi$ for all sufficiently
large $n$. Combining these bounds with the ordering of the
$L_i^n$ gives
\[
\frac{L_i^n\ell(s)}{L_{m-1}^n\bar\ell}\leq\xi
\quad\text{for }(i,s)\neq(m-1,h),
\]
and
\[
\frac{L_2^n\underline\ell}{L_i^n\ell(s)}\leq\xi
\quad\text{for }(i,s)\neq(2,l).
\]
When $m=3$, there is only one transient state, and these
inequalities follow directly from the signal likelihood-ratio
bounds.

Let
\[
w_n^+:=
\frac{
\nu_{m-1}^L(\Pi^n)\pi^L(h)f^n(m-1,h)(m)
}{
\sum_{i\in T}\sum_{s\in S}
\nu_i^L(\Pi^n)\pi^L(s)f^n(i,s)(m)
},
\]
and
\[
w_n^-:=
\frac{
\nu_2^H(\Pi^n)\pi^H(l)f^n(2,l)(1)
}{
\sum_{i\in T}\sum_{s\in S}
\nu_i^H(\Pi^n)\pi^H(s)f^n(i,s)(1)
}.
\]
These are the shares of high absorption under $L$ and low
absorption under $H$ contributed by their respective extreme
state--signal pairs.

The two factors in \eqref{eqn:absorbing_prob} are weighted
averages of $L_i^n\ell(s)$ and its reciprocal, respectively.
The preceding uniform bounds therefore imply
\[
\begin{aligned}
R_n
&\leq
\gamma\frac{L_{m-1}^n}{L_2^n}
\bigl[w_n^++\xi(1-w_n^+)\bigr]
\bigl[w_n^-+\xi(1-w_n^-)\bigr]\\
&\leq
\gamma^{m-2}
\bigl[w_n^++\xi(1-w_n^+)\bigr]
\bigl[w_n^-+\xi(1-w_n^-)\bigr].
\end{aligned}
\]
Each bracketed factor is at most one. Since
$R_n\to\gamma^{m-2}$ and $\xi<1$, both factors must converge
to one. Thus,
\[
w_n^+\rightarrow1,
\qquad
w_n^-\rightarrow1.
\]

Finally, the pair $(m-1,h)$ maximizes the $H$-to-$L$ ratio
of high-absorption flows. Its share under $H$ is therefore
at least its share $w_n^+$ under $L$. Similarly, $(2,l)$
maximizes the $L$-to-$H$ ratio of low-absorption flows,
so its share under $L$ is at least its share $w_n^-$ under
$H$. Consequently, both concentration statements hold under
either state of nature. This establishes condition
\textup{(i)} and completes the proof.
\end{proof}

To establish part \textup{(ii)} of
\Cref{thm:behavioral_implication}, the following lemma rules out
a non-vanishing transition into an absorbing state. Cases 1--4
concern absorption from a non-extreme transient state, into the
opposite absorbing state, or upon a signal other than the
corresponding extreme signal. Cases 5--6 address the remaining
transitions, from state $2$ to state $1$ upon $l$ and from state
$m-1$ to state $m$ upon $h$, distinguishing whether the source
state coincides with the initial state.

\begin{lemma}\label{lem:part_ii_thm}
Let $\{\Pi^n\}_{n=1}^\infty$ be a sequence of parsimonious
$m$-state protocols, where $m\geq4$, with a common deterministic
initial state $i_0\in\{2,\ldots,m-1\}$. Suppose there exist
$\epsilon>0$, $i^*\in\{2,\ldots,m-1\}$, $j^*\in\{1,m\}$, and
$s^*\in S$ such that
\(
f^n(i^*,s^*)(j^*)>\epsilon
\) for every $n$. 
Then the following statements hold.

    \begin{enumerate}
        \item If $i^*\notin \{2,m-1\}$ and $i_0\neq i^*$, then the optimal absorption, the optimal mixing and the optimal bias conditions cannot all hold.

        \item If $i^*\notin \{2,m-1\}$ and $i_0=i^*$, then optimal absorption condition cannot hold.

        \item If $(i^*,j^*)=(2,m) \text{ or } (m-1,1)$, then the optimal absorption and the optimal bias conditions cannot both hold.

        \item If $(i^*,j^*)=(2,1)$, $s^*\neq l$, or $(i^*,j^*)=(m-1,m)$, $s^*\neq h$, then the optimal absorption condition cannot hold.

        \item If $(i^*,j^*,s^*)=(2,1,l) \text{ or } (m-1,m,h)$ and $i_0\neq i^*$, then optimal mixing condition cannot hold.

        \item If $(i^*,j^*,s^*)=(2,1,l) \text{ or } (m-1,m,h)$ and $i_0= i^*$, then the optimal absorption, the optimal mixing, and the optimal bias conditions cannot all hold.
    \end{enumerate}
\end{lemma}
\begin{proof}
    See \Cref{proof:lem:part_ii_thm}.
\end{proof}

\begin{proof}[Proof of \Cref{thm:behavioral_implication}]
Let $\{\Pi^n\}_{n=1}^\infty$ satisfy the assumptions of the theorem.
By \Cref{lem:3-conditions}, this sequence satisfies the optimal
absorption, optimal mixing, and optimal bias conditions.
Part \textup{(i)} follows from the optimal absorption condition
and the absorption formula \eqref{lem:absorbing-formula}.

We first prove part \textup{(ii)} for sequences with deterministic
initialization. Suppose, to the contrary, that the conclusion fails.  Since the memory and signal sets are finite,
there exist $\epsilon>0$, a subsequence $\{\Pi^{n_k}\}_{k=1}^\infty$,
and fixed $i^*\in\{2,\ldots,m-1\}$, $j^*\in\{1,m\}$, and
$s^*\in S$ such that
\[
f^{n_k}(i^*,s^*)(j^*)>\epsilon
\qquad\text{for every }k.
\]
For all sufficiently large $k$, the initial state is transient:
starting in an absorbing state yields at most $\bar p$, whereas
the optimal value is strictly greater than $\bar p$.
Passing to a further subsequence, we may therefore take the
initial state to be constant:
\[
i_0(\Pi^{n_k})=i_0\in\{2,\ldots,m-1\}
\qquad\text{for every }k.
\]
For notational simplicity, relabel this subsequence as
$\{\Pi^n\}_{n=1}^\infty$.

The subsequence remains asymptotically optimal and hence satisfies
all three conditions in \Cref{lem:3-conditions}. However, the six
cases in \Cref{lem:part_ii_thm} exhaust the possible configurations
of $(i_0,i^*,j^*,s^*)$, and each implies that at least one of these
conditions fails. This is a contradiction.

Therefore, $f^n(i,s)(j)\to0$ for every
$i\in\{2,\ldots,m-1\}$, $s\in S$, and $j\in\{1,m\}$. In particular,
\[
f^n(i,s)(\{1,m\})
=
f^n(i,s)(1)+f^n(i,s)(m)
\rightarrow0,
\]
which establishes part \textup{(ii)} for sequences with deterministic
initialization.

For arbitrary initial distributions, write $\Pi^n=(f^n,g^n,a^n)$
and let $\Pi_j^n=(f^n,\delta_j,a^n)$, where $\delta_j$ assigns
probability one to state $j$. Choose
$j_n\in\arg\max_{j\in M}U^R(\Pi_j^n)$.
Since the receiver's payoff under a parsimonious protocol is linear
in the initial distribution,
\[
U^R(\Pi^n)
=
\sum_{j\in M}g^n(j)U^R(\Pi_j^n)
\leq
U^R(\Pi_{j_n}^n)
\leq
\sup_\Pi U^R(\Pi).
\]
Thus, $\{\Pi_{j_n}^n\}$ is also asymptotically optimal. Applying
the deterministic-initialization case just proved gives the desired
conclusion for its transition rules, which are unchanged from those
of $\Pi^n$. Although changing the initialization may change the
likelihood-ratio ordering of transient states, part \textup{(ii)}
is invariant under relabeling because it concerns transitions from
every transient state. This completes the proof.
\end{proof}

\printbibliography

\newpage
\clearpage

\pagenumbering{arabic}
\renewcommand{\thepage}{OA-\arabic{page}}

\phantomsection
\section*{\centering \huge Online Appendix}
\addcontentsline{toc}{section}{Online Appendix}
\etocsettocdepth.toc{section}

\section{Asymmetric Receiver Payoffs}\label{sec:Asymmetric Receiver Payoff}
The assumption that correct decisions yield the same payoff in both states is irrelevant, up to a renormalization. Suppose the receiver earns $1$ from choosing $H$ in state $H$ and $\zeta>0$ from choosing $L$ in state $L$, while mismatches and non-termination continue to yield zero. The sender's payoffs and information remain unchanged. Define 
\[ c:=p+\zeta(1-p), \qquad \widetilde p:=\frac{p}{p+\zeta(1-p)}. \] 
For any strategy profile, let $\alpha$ and $\beta$ denote the probabilities of a correct terminal decision conditional on $H$ and $L$, respectively. The receiver's expected payoff is 
\[ p\alpha+\zeta(1-p)\beta = c\bigl[\widetilde p\,\alpha+(1-\widetilde p)\beta\bigr]. \] 
Thus, his objective is a positive multiple of the baseline objective with prior $\widetilde p$. 

The sender's best responses to any given protocol are unaffected. Consequently, the receiver's design problem is equivalent to the baseline problem with the payoff-adjusted prior $\widetilde p$. The structural results and the behavioral implications for asymptotically optimal parsimonious protocols therefore extend with $p$ replaced by $\widetilde p$ in the relevant parameter restrictions. 

The receiver's optimal value is obtained from \Cref{thm:supUR} by the same substitution and multiplication by $c$. In particular, learning is strictly valuable if and only if $\gamma^{m-2}>\kappa_\zeta$, where 
\[ \kappa_\zeta := \max\left\{ \frac{\tilde{p}}{1-\tilde{p}}, \frac{1-\tilde{p}}{\tilde{p}} \right\}=\max\left\{ \frac{p}{\zeta(1-p)}, \frac{\zeta(1-p)}{p} \right\}. \] 
The sender's ex ante payoff, however, must still be evaluated using the actual prior $p$, rather than $\widetilde p$.

\section{The Sender's Optimal Payoff}\label{sec:sender optimal value}
\subsection{Results}

Fix the memory bound $m\geq2$. For a general protocol $\Pi$, write
\[
\overline U^R(\Pi)
:=\sup_{\sigma\in\br(\Pi)}U^R(\Pi,\sigma),
\qquad
U^S(\Pi):=\max_\sigma U^S(\Pi,\sigma),
\]
and let $V_m:=\sup_\Pi\overline U^R(\Pi)$, where the supremum is
restricted to protocols with at most $m$ memory states. By
\Cref{thm:parsimonious,thm:supUR}, $V_m$ is the receiver's optimal
value characterized above. The sender's payoff is the same under
all her best responses; the receiver's payoff need not be.

\begin{theorem}\label{thm:sender_payoff_general}
Let $\{\Pi^n\}_{n=1}^\infty$ be a sequence of protocols with at
most $m$ memory states satisfying
$\overline U^R(\Pi^n)\to V_m$. If either
$\gamma^{m-2}>\kappa$ or $p\neq\tfrac12$, then the sender's
optimal payoff converges, with
\begin{equation}\label{eq:sender-payoff-limit}
\lim_{n\to\infty}U^S(\Pi^n)
=
\begin{cases}
\displaystyle
p+\frac{2p-1}{\gamma^{m-2}-1},
& \textup{if }\gamma^{m-2}>\kappa,\\[0.6em]
1,
& \textup{if }\gamma^{m-2}\leq\kappa
  \textup{ and }p>\tfrac12,\\[0.3em]
0,
& \textup{if }\gamma^{m-2}\leq\kappa
  \textup{ and }p<\tfrac12.
\end{cases}
\end{equation}
\end{theorem}

In the learning region, the limiting payoff lies strictly between
zero and one. The excluded case is $m=2$ and $p=\tfrac12$.
In that case, the receiver can attain a payoff of $\tfrac12$ by choosing
between two absorbing states at initialization. Assigning
probability $z$ to the $H$-decision state gives the sender payoff
$z$, for any $z\in[0,1]$. Thus, even a sequence of exactly optimal
protocols need not generate a convergent sender payoff.

The theorem also describes the sender's preferences over signal
informativeness when the receiver reoptimizes his protocol. The
sender prefers greater informativeness when the prior disfavors
her preferred action $H$, and lower informativeness when the
prior favors it.

\begin{corollary}\label{cor:cs-sender}
Suppose $\gamma^{m-2}>\kappa$, and let $W_m(p,\gamma)$ denote
the common limiting sender payoff in
\Cref{thm:sender_payoff_general}. Holding $p$ and $m$ fixed:

\textup{(i)} If $p<\tfrac12$, then $W_m(p,\gamma)$ is strictly
increasing in $\gamma$. It converges to $p$ as $\gamma\to\infty$
and to $0$ as $\gamma^{m-2}\downarrow\kappa$.

\textup{(ii)} If $p>\tfrac12$, then $W_m(p,\gamma)$ is strictly
decreasing in $\gamma$. It converges to $p$ as $\gamma\to\infty$
and to $1$ as $\gamma^{m-2}\downarrow\kappa$.

\textup{(iii)} If $p=\tfrac12$, then $W_m(p,\gamma)=\tfrac12$.
\end{corollary}

\begin{proof}
The claims follow directly from \eqref{eq:sender-payoff-limit}.
In particular,
\[
\frac{\partial W_m(p,\gamma)}{\partial\gamma}
=
-\frac{(2p-1)(m-2)\gamma^{m-3}}
       {(\gamma^{m-2}-1)^2}.
\]
The learning condition implies $m>2$.
\end{proof}

In particular, the symmetric binary case has the following simple
implication.

\begin{corollary}\label{cor:sender-payoff-symmetric}
In the symmetric binary case, if $m\geq3$ and
$\overline U^R(\Pi^n)\to V_m$, then
\[
\lim_{n\to\infty}U^S(\Pi^n)=\frac12.
\]
When $m=2$, the receiver's optimality does not determine the
sender's payoff, as described above.
\end{corollary}

\subsection{Proof of \Cref{thm:sender_payoff_general}}

We first compute the sender's limiting payoff for parsimonious
protocols. To extend the result to arbitrary protocols, we show
that a terminating sender best response always exists. The two
state-contingent probabilities of action $H$ then determine both
the sender's value and the receiver's best payoff across best
responses. Finally, we compare each protocol's performance under
nearby priors. Because the receiver's optimal value is
differentiable in the cases covered by the theorem, asymptotic
optimality uniquely determines those state-contingent probabilities.
This argument also covers the boundary
$\gamma^{m-2}=\kappa$ when $p\neq\tfrac12$.

\begin{lemma}\label{lem:sender_payoff}
Suppose either $\gamma^{m-2}>\kappa$ or $p\neq\tfrac12$.
For every sequence of parsimonious $m$-state protocols satisfying
$U^R(\Pi^n)\to V_m$, the sender's payoff converges to the value
in \eqref{eq:sender-payoff-limit}.
\end{lemma}

\begin{proof}
The absorption-probability pair
$(\mu_m^H(\Pi^n),\mu_1^L(\Pi^n))$ is feasible for
\eqref{eqn:receiver_opt_relax}. By tightness of the relaxation,
its limit points solve that problem. The optimizer is unique in
the cases under consideration by
\Cref{lem:auxiliary-optimization}. Denoting this optimizer by
$(\alpha^*,\beta^*)$, including its corner solutions in the
no-learning region, we obtain
\[
\mu_m^H(\Pi^n)\to\alpha^*,
\qquad
\mu_1^L(\Pi^n)\to\beta^*.
\]
Since absorption occurs almost surely,
\[
U^S(\Pi^n)
=p\mu_m^H(\Pi^n)+(1-p)(1-\mu_1^L(\Pi^n)),
\]
and hence
\begin{equation}\label{eqn:seller_value}
\lim_{n\to\infty}U^S(\Pi^n)
=p\alpha^*+(1-p)(1-\beta^*).
\end{equation}
In the learning region, substitution of the optimizer gives
\[
p\alpha^*+(1-p)(1-\beta^*)
=\frac{p\gamma^{m-2}-(1-p)}{\gamma^{m-2}-1}
=p+\frac{2p-1}{\gamma^{m-2}-1}.
\]
In the no-learning region, the optimizer is $(1,0)$ for
$p>\tfrac12$ and $(0,1)$ for $p<\tfrac12$, yielding the other
two cases.
\end{proof}

For the general case, we will vary the prior while holding the
protocol and signal distributions fixed. We make this dependence
explicit by writing $\overline U^R(\Pi;p)$, $U^S(\Pi;p)$,
and $V_m(p)$.

\begin{lemma}\label{lem:sender-terminating-response}
For any protocol $\Pi$, there exists a pure stationary sender
best response under which the game terminates in finite time
almost surely, conditional on either state of nature.

Let $z^\theta(\Pi)$ be the sender's maximal conditional
probability of inducing action $H$ when the state of nature is
$\theta$. These numbers and the sender's best-response set do
not depend on the prior $p\in(0,1)$. Moreover,
\begin{equation}\label{eq:sender-payoff-representation}
\begin{aligned}
\overline U^R(\Pi;p)
&=p z^H(\Pi)+(1-p)(1-z^L(\Pi)),\\
U^S(\Pi;p)
&=p z^H(\Pi)+(1-p)z^L(\Pi).
\end{aligned}
\end{equation}
\end{lemma}

\begin{proof}
Take a pure stationary sender best response, whose existence was
established for the finite stopping problem. Fix $\theta$.
On any path that does not terminate, the induced finite-state
chain eventually enters a closed recurrent communicating class
in which the sender always continues and the protocol never
terminates autonomously.

Consider such a class that is reached with positive probability.
If it contained a state $i$ with $a(i)>0$, the sender could improve
her conditional payoff by stopping at $i$: after entering the
class, state $i$ is reached almost surely, whereas perpetual
continuation yields zero. This contradicts optimality. Thus every
state in the class has $a(i)=0$. Changing the sender's strategy
to stop upon entering each such class preserves her payoff,
since it replaces a zero-payoff nonterminating path with action
$L$, which also gives her zero. After this modification, the game
terminates almost surely. Applying the argument under each
$\theta$ yields a pure stationary best response with the claimed
property.

The sender knows $\theta$ and can select her strategy separately
under the two states of nature. Her state-contingent stopping
problems therefore do not depend on $p$. Since both prior
probabilities are positive, a strategy is an ex ante best response
if and only if its two components are conditional best responses.
Thus every best response induces $H$ with probability
$z^\theta(\Pi)$ conditional on $\theta$.

Under any best response, the receiver's payoff conditional on $H$
is $z^H(\Pi)$. Conditional on $L$, his payoff is at most
$1-z^L(\Pi)$, because nontermination also yields zero. The
terminating best response attains this upper bound. This proves
both identities in \eqref{eq:sender-payoff-representation}.
\end{proof}

\begin{proof}[Proof of \Cref{thm:sender_payoff_general}]
Fix the prior $p$ in one of the cases covered by the theorem.
For each $n$, set
\[
x_n:=z^H(\Pi^n),
\qquad
y_n:=1-z^L(\Pi^n).
\]
By \Cref{lem:sender-terminating-response}, for every prior
$u\in(0,1)$,
\[
\overline U^R(\Pi^n;u)=u x_n+(1-u)y_n\leq V_m(u),
\]
and
\[
U^S(\Pi^n;p)=p x_n+(1-p)(1-y_n).
\]
The same pair $(x_n,y_n)$ applies at every prior because we
hold $\Pi^n$ fixed and the sender's conditional best responses
are independent of the prior.

Let $R_n:=\overline U^R(\Pi^n;p)\to V_m(p)$.
For $0<\delta<\min\{p,1-p\}$, evaluating the preceding
inequality at $p-\delta$ and $p+\delta$ gives
\[
\frac{R_n-V_m(p-\delta)}{\delta}
\leq x_n-y_n
\leq\frac{V_m(p+\delta)-R_n}{\delta}.
\]
Taking lower and upper limits in $n$, and then letting
$\delta\downarrow0$, shows that
\begin{equation}\label{eq:sender-prior-slope}
x_n-y_n\rightarrow V_m'(p),
\end{equation}
provided $V_m$ is differentiable at $p$.

The value formula in \Cref{thm:supUR} is differentiable at
every $p\in(0,1)$ except when $m=2$ and $p=\tfrac12$.
Specifically,
\[
V_m'(p)=
\begin{cases}
\displaystyle
\frac{\gamma^{(m-2)/2}(2p-1)}
{(\gamma^{m-2}-1)\sqrt{p(1-p)}},
&\text{if }\gamma^{m-2}>\kappa,\\[0.9em]
1,
&\text{if }\gamma^{m-2}\leq\kappa
 \text{ and }p>\tfrac12,\\[0.3em]
-1,
&\text{if }\gamma^{m-2}\leq\kappa
 \text{ and }p<\tfrac12.
\end{cases}
\]
At either boundary of the learning region, the derivative of
the interior formula equals that of the adjoining prior-only
branch. Thus \eqref{eq:sender-prior-slope} applies in every
case covered by the theorem, including the boundaries.

Since $R_n=p x_n+(1-p)y_n$, we have
\[
x_n=R_n+(1-p)(x_n-y_n),
\qquad
y_n=R_n-p(x_n-y_n).
\]
Consequently,
\begin{equation}\label{eq:sender-terminal-action-limits}
x_n\rightarrow V_m(p)+(1-p)V_m'(p),
\qquad
y_n\rightarrow V_m(p)-pV_m'(p).
\end{equation}
In the learning region, substituting the value and derivative
shows that these limits are $\alpha^*$ and $\beta^*$,
respectively. The sender's payoff therefore converges to
\[
p\alpha^*+(1-p)(1-\beta^*)
=p+\frac{2p-1}{\gamma^{m-2}-1},
\]
as in \Cref{lem:sender_payoff}.

In the no-learning region with $p>\tfrac12$, we have
$V_m(p)=p$ and $V_m'(p)=1$, so $(x_n,y_n)\to(1,0)$ and
$U^S(\Pi^n;p)\to1$. When $p<\tfrac12$, we have
$V_m(p)=1-p$ and $V_m'(p)=-1$, giving
$(x_n,y_n)\to(0,1)$ and $U^S(\Pi^n;p)\to0$.
This establishes all three cases.

The economic logic of the prior comparison is that a near-optimal
protocol cannot have an arbitrarily different tradeoff between
its two state-contingent outcomes. If its slope in the prior
remained different from $V_m'(p)$, the same protocol would
outperform the optimal value at a nearby prior. Differentiability
therefore pins down the limiting outcomes and, in turn, the
sender's payoff.
\end{proof}

\section{Proof of \Cref{lem:part_ii_thm}}\label{proof:lem:part_ii_thm}
\begin{proof}

\begin{enumerate}

    \item Without loss of generality, assume $j^*=1$.     If $i_0>i^*$, then by the optimal absorption condition,

    \begin{equation*}
        \frac{\nu^\theta_{i^*} (\Pi^n) \pi^\theta(s^*) f^{n}(i^*,s^*)(1)}{\nu^\theta_2 (\Pi^n) \pi^\theta(l) f^{n}(2,l)(1)}\rightarrow 0.
    \end{equation*}

    Since $f^{n}(i^*,s^*)(1)>\epsilon$ for all $n$,

    \begin{equation*}
        \frac{\nu^\theta_{i^*} (\Pi^n)}{\nu^\theta_2 (\Pi^n) \pi^\theta(l) f^{n}(2,l)(1)}\rightarrow 0.
    \end{equation*}
    Therefore,
    \begin{equation*}
    \begin{aligned}
        \frac{\sum_{j=i^*+1}^{m-1}\sum_{s\in S}\nu^\theta_{i^*}(\Pi^n)\pi^\theta(s)\tilde{f}^n(i^*,s)(j)}{\sum_{i=2}^{i^*}\sum_{j=i^*+1}^{m-1}\sum_{s\in S}\nu^\theta_i(\Pi^n)\pi^\theta(s)\tilde{f}^n(i,s)(j)}&\leq \frac{\nu^\theta_{i^*} (\Pi^n)}{\nu^\theta_2 (\Pi^n) \pi^\theta(l) f^{n}(2,l)(1)} \\
        &\rightarrow 0,
    \end{aligned}    
    \end{equation*}
contradicting the optimal mixing condition (ii.a).

    If $i_0<i^*$, similar to the previous case, we have
   \begin{equation}\label{eqn:stationary_ratio}
        \frac{\nu^\theta_{i^*} (\Pi^n)}{\nu^\theta_2 (\Pi^n) \pi^\theta(l) f^{n}(2,l)(1)}\rightarrow 0.
    \end{equation}
    By the optimal bias condition (or that the receiver's payoff tends to his optimal payoff), we have
    \begin{equation*}
        \frac{\mu^H_m(\Pi^n)}{\mu^H_1(\Pi^n)}\rightarrow \frac{\alpha^*}{1-\alpha^*};
    \end{equation*}
    \begin{equation*}
        \frac{\mu^L_m(\Pi^n)}{\mu^L_1(\Pi^n)}\rightarrow \frac{1-\beta^*}{\beta^*}.
    \end{equation*}

    By \eqref{lem:absorbing-formula}, 
    \begin{equation*}
        \frac{\mu^\theta_m(\Pi^n)}{\mu^\theta_1(\Pi^n)}=\frac{\sum_{i=2}^{m-1} \sum_{s\in S}\nu^\theta_i(\Pi^n)\pi^\theta(s)f^n(i,s)(m)}{\sum_{i=2}^{m-1} \sum_{s\in S}\nu^\theta_i(\Pi^n)\pi^\theta(s)f^n(i,s)(1)}.
    \end{equation*}
    By the optimal absorption condition,
    \begin{equation*}
        \frac{\nu^H_{m-1} (\Pi^n) \pi^H(h) f^{n}(m-1,h)(m)}{\nu^H_2 (\Pi^n) \pi^H(l) f^{n}(2,l)(1)}\rightarrow \frac{\alpha^*}{1-\alpha^*};
    \end{equation*}
    \begin{equation*}
        \frac{\nu^L_{m-1} (\Pi^n) \pi^L(h) f^{n}(m-1,h)(m)}{\nu^L_2 (\Pi^n) \pi^L(l) f^{n}(2,l)(1)}\rightarrow \frac{1-\beta^*}{\beta^*}.
    \end{equation*}
Together with \Cref{eqn:stationary_ratio}, we have
    \begin{equation*}
        \frac{\nu^\theta_{i^*} (\Pi^n)}{\nu^\theta_{m-1} (\Pi^n) \pi^\theta(h) f^{n}(m-1,h)(m)}\rightarrow 0.
    \end{equation*}
Similar to the previous case,
    \begin{equation*}
    \begin{aligned}
        \frac{\sum_{j=2}^{i^*-1}\sum_{s\in S}\nu^\theta_{i^*}(\Pi^n)\pi^\theta(s)\tilde{f}^n(i^*,s)(j)}{\sum_{i=i^*}^{m-1}\sum_{j=2}^{i^*-1}\sum_{s\in S}\nu^\theta_i(\Pi^n)\pi^\theta(s)\tilde{f}^n(i,s)(j)}&\leq \frac{\nu^\theta_{i^*} (\Pi^n)}{\nu^\theta_{m-1} (\Pi^n) \pi^\theta(h) f^{n}(m-1,h)(m)} \\
        &\rightarrow 0,
    \end{aligned}    
    \end{equation*} 
violating the optimal mixing condition (ii.a).

     \item By considering the probability of being absorbed in the first period,
    \begin{equation*}
        \textup{P}^\theta(m_{\tau^\theta(\Pi^n)-1}=2, s_ {\tau^\theta(\Pi^n)-1}=l | m_{\tau^\theta(\Pi^n)}=1)< 1-\epsilon \pi^\theta(s^*)
    \end{equation*}

    since the probability of receiving signal $s^*$ and being absorbed immediately is greater than $\epsilon \pi^\theta(s^*)$. This violates the optimal absorption condition.

    \item Without loss of generality, we only prove the case $(i^*,j^*)=(2,m)$.
 By the optimal absorption condition,
    \begin{equation*}
        \frac{\nu^H_{2} (\Pi^n) \pi^H(s^*) f^{n}(2,s^*)(m)}{\nu^H_{m-1} (\Pi^n) \pi^H(h) f^{n}(m-1,h)(m)}\rightarrow 0.
    \end{equation*}
However, since $f^n(2,s^*)(m)>\epsilon$, we have
    \begin{align*}
        \frac{\nu^H_{2} (\Pi^n) \pi^H(l) f^{n}(2,l)(1)}{\nu^H_{m-1} (\Pi^n) \pi^H(h) f^{n}(m-1,h)(m)}&= \frac{\pi^H(l) f^{n}(2,l)(1)}{\pi^H(s^*) f^{n}(2,s^*)(m)}\frac{\nu^H_{2} (\Pi^n) \pi^H(s^*) f^{n}(2,s^*)(m)}{\nu^H_{m-1} (\Pi^n) \pi^H(h) f^{n}(m-1,h)(m)}\\
        &\leq \frac{\pi^H(l)}{\pi^H(s^*)\epsilon}\frac{\nu^H_{2} (\Pi^n) \pi^H(s^*) f^{n}(2,s^*)(m)}{\nu^H_{m-1} (\Pi^n) \pi^H(h) f^{n}(m-1,h)(m)}\\
        &\rightarrow 0.
    \end{align*} 
By optimal absorption and \eqref{lem:absorbing-formula}, it follows that $\frac{\mu^H_1(\Pi^n)}{\mu^H_m(\Pi^n)}\rightarrow 0,$   contradicting the optimal bias condition.

    \item  Without loss of generality, we only prove the case $(i^*,j^*)=(2,1)$, $s^*\neq l$. Notice that
    \begin{equation*}
    \begin{aligned}
        \frac{\nu^\theta_2(\Pi^n)\pi^\theta(l)f^n(2,l)(1)}{\sum_{i=2}^{m-1} \sum_{s\in S}\nu^\theta_i(\Pi^n)\pi^\theta(s)f^n(i,s)(1)}&\leq \frac{\nu^\theta_2(\Pi^n)\pi^\theta(l)f^n(2,l)(1)}{\nu^\theta_2(\Pi^n)\pi^\theta(l)f^n(2,l)(1)+\nu^\theta_2(\Pi^n)\pi^\theta(s^*)f^n(2,s^*)(1)}\\
        &\leq \frac{\pi^\theta(l)f^n(2,l)(1)}{\pi^\theta(l)f^n(2,l)(1)+\pi^\theta(s^*)\epsilon}\\
        &\leq \frac{\pi^\theta(l)}{\pi^\theta(l)+\pi^\theta(s^*)\epsilon}\\
        &<1.
    \end{aligned}
    \end{equation*}
   So the optimal absorption condition is violated.

    \item Without loss of generality, we only prove the case $(i^*,j^*,s^*)=(2,1,l)$. Notice that
    \begin{equation*}
    \begin{aligned}
        \frac{\sum_{j=3}^{m-1}\sum_{s\in S} \pi^H(s)\tilde{f}^n(2,s)(j)}{\sum_{j=3}^{m-1}\sum_{s\in S} \pi^L(s)\tilde{f}^n(2,s)(j)}&\leq \frac{\pi^H(l)f(2,l)(1)+\bar{\ell}}{\pi^L(l)f(2,l)(1)+1}\\
        &\leq \frac{\pi^H(l)\epsilon+\bar{\ell}}{\pi^L(l)\epsilon+1}\\
        &<\bar{\ell},
    \end{aligned}
    \end{equation*}
 violating the optimal mixing condition (ii.b).

    \item Without loss of generality, we only prove the case $(i^*,j^*,s^*)=(2,1,l)$, $i_0=i^*$. By the optimal bias and the optimal absorption conditions,
    \begin{equation*}
        \frac{\nu^H_{m-1} (\Pi^n)\pi^H(h)f^n(m-1,h)(m)}{\nu^H_{2}(\Pi^n)\pi^H(l)f^n(2,l)(1)}\rightarrow \frac{\alpha^*}{1-\alpha^*}.
    \end{equation*}
Therefore, for large enough $n$,
    \begin{equation*}
        \frac{\nu^H_{m-1}(\Pi^n)\pi^H(h)f^n(m-1,h)(m)}{\nu^H_{2}(\Pi^n)\pi^H(l)f^n(2,l)(1)}> \frac{\alpha^*}{2(1-\alpha^*)}.
    \end{equation*}
Now consider the balance equation for ${T}(\Pi^n)$ between sets $\{2\}$ and $\{3,...,m-1\}$:
    \begin{equation*}
        \nu^H_2 (\Pi^n)\sum_{i=3}^{m-1} \sum_{s\in S} \pi^H(s)f^n(2,s)(i)= \sum_{i=3}^{m-1} \sum_{s\in S} \nu^H_i(\Pi^n) \pi^H(s)\tilde{f}^n(i,s)(2).
    \end{equation*}
So for large enough $n$,
    \begin{align*}
        \frac{\nu^H_{m-1}(\Pi^n)\pi^H(h)f^n(m-1,h)(m)}{\sum_{i=3}^{m-1} \sum_{s\in S} \nu^H_i(\Pi^n) \pi^H(s)\tilde{f}^n(i,s)(2)}&=\frac{\nu^H_{m-1}(\Pi^n)\pi^H(h)f^n(m-1,h)(m)}{\nu^H_2(\Pi^n) \sum_{i=3}^{m-1} \sum_{s\in S} \pi^H(s)f^n(2,s)(i)}\\
        &\geq \frac{\nu^H_{m-1}(\Pi^n)\pi^H(h)f^n(m-1,h)(m)}{\nu^H_2(\Pi^n)}\\
        &=\frac{\nu^H_{m-1}(\Pi^n)\pi^H(h)f^n(m-1,h)(m)}{\nu^H_2(\Pi^n) \pi^H(l)f^n(2,l)(1)} \pi^H(l)f^n(2,l)(1)\\
        &> \epsilon \pi^H(l)\frac{\alpha^*}{2(1-\alpha^*)}.
    \end{align*}

            Set
\[
c:=
\min\left\{
\frac12,\,
\epsilon\pi^H(l)\frac{\alpha^*}{2(1-\alpha^*)}
\right\}\in(0,1).
\]
The preceding flow ratio is greater than $c$ for all
sufficiently large $n$.

If $m\geq5$, then $m-1\neq3$. Since $i_0=2$, the transition
from $m-1$ into the absorbing state $m$ is redirected to $2$
in the modified chain. This flow is separate from the flow
originating at state $3$. Hence
\[
\frac{
\nu_3^H(\Pi^n)
\sum_{s\in S}\pi^H(s)\widetilde f^n(3,s)(2)
}{
\sum_{i=3}^{m-1}\sum_{s\in S}
\nu_i^H(\Pi^n)\pi^H(s)\widetilde f^n(i,s)(2)
}
<1-c.
\]
This contradicts optimal mixing condition \textup{(ii.a)}
under $\theta=H$.

If $m=4$, define
\[
d_n^\theta:=
\sum_{s\in S}\pi^\theta(s)\widetilde f^n(3,s)(2),
\qquad
w_n^\theta:=
\frac{\pi^\theta(h)f^n(3,h)(4)}{d_n^\theta}.
\]
Here $w_n^\theta$ is the share of the modified $3\to2$
transition probability arising from an $h$-triggered absorption
in state $4$ followed by a reset to $2$. The preceding bound
gives $w_n^H>c$ for all sufficiently large $n$.

Because $\pi^H(s)\geq\underline\ell\,\pi^L(s)$ for every $s$,
we have $d_n^H/d_n^L\geq\underline\ell$. Therefore,
\[
w_n^L
=
w_n^H\frac{d_n^H}{\bar\ell\,d_n^L}
\geq
w_n^H\frac{\underline\ell}{\bar\ell}
>
\frac{c}{\gamma}.
\]

The reset component has likelihood ratio $\bar\ell$, while
every remaining component has likelihood ratio at least
$\underline\ell$. Averaging with the $L$-transition
probabilities as weights gives
\[
\begin{aligned}
\frac{d_n^H}{d_n^L}
&\geq
w_n^L\bar\ell+(1-w_n^L)\underline\ell\\
&>
\underline\ell
+\frac{c}{\gamma}(\bar\ell-\underline\ell)\\
&>\underline\ell.
\end{aligned}
\]
The gap above $\underline\ell$ is independent of $n$,
contradicting optimal mixing condition \textup{(ii.b)}.
\end{enumerate}
\end{proof}

\section{Proof of \Cref{thm:polarization}}\label{proof:thm:polarization}

We use the labeling of transient states from
\Cref{thm:behavioral_implication}, so that states $2$
and $m-1$ have the lowest and highest transient-state
likelihood ratios, respectively. 
For each transient state $i$, let $\chi_i^\theta(\Pi)$
denote the probability of eventual absorption in $m$
under $\theta$ when the protocol is initialized at $i$,
with its transition and action rules unchanged.
The sender continues providing information until absorption.

\begin{lemma}\label{lem:continuation-nondegeneracy} Suppose the assumptions of \Cref{thm:polarization} hold. Then, for each $\theta\in\{H,L\}$, \[ \liminf_{n\to\infty}\chi_2^\theta(\Pi^n)>0 \] and \[ \liminf_{n\to\infty} \bigl(1-\chi_{m-1}^\theta(\Pi^n)\bigr)>0. \] \end{lemma}

\begin{proof} 
We first prove the claim for memory state $2$. Suppose, to the contrary, that along a subsequence \[ \chi_2^\theta(\Pi^n)\to0 \] for some $\theta\in\{H,L\}$. We continue to index this subsequence by $n$, and let $\theta'$ denote the other state of nature. Apply \Cref{lem:spread} to the protocol initialized at
state $2$. Applying the same bound after interchanging
the roles of the two absorbing states gives the reverse
odds-ratio bound. The action rule can be reassigned so
that the relabeled protocol remains parsimonious, without
changing its absorption probabilities. Consequently,
\[ \bigl(1-\chi_2^\theta(\Pi^n)\bigr) \chi_2^{\theta'}(\Pi^n) \leq \gamma^{m-2} \chi_2^\theta(\Pi^n) \bigl(1-\chi_2^{\theta'}(\Pi^n)\bigr). \] 
It follows that $\chi_2^{\theta'}(\Pi^n)\to0. $ Hence, \[ \chi_2^H(\Pi^n)\to0 \qquad\text{and}\qquad \chi_2^L(\Pi^n)\to0. \] 

Construct an $(m-1)$-state parsimonious protocol $\widehat\Pi^n$ by deleting state $1$, making state $2$ the low-action absorbing state, and redirecting every initial or transition probability assigned to state $1$ to state $2$. Couple $\Pi^n$ and $\widehat\Pi^n$ so that they use the same signals and transition randomization until $\Pi^n$ first reaches one of the states $\{1,2,m\}$. The two protocols induce the same terminal action unless $\Pi^n$ first reaches state $2$ and subsequently becomes absorbed in state $m$. Therefore, 
\begin{equation}\label{eqn:convergence} \left| U^R(\Pi^n)-U^R(\widehat\Pi^n) \right| \leq p\,\chi_2^H(\Pi^n) + (1-p)\chi_2^L(\Pi^n) \to0. \end{equation}

Let $ V_{m-1}$ be the receiver's optimal payoff when he only has $m-1$ memory states. By \Cref{thm:supUR}, under the assumptions of \Cref{thm:polarization}, 
\begin{equation}\label{eqn:monotone-payoff}\sup_{\Pi}U^R(\Pi)>V_{m-1}.\end{equation} 
It follows from \eqref{eqn:convergence} that \[ \sup_{\Pi}U^R(\Pi)= \lim_{n\to\infty}U^R(\Pi^n)=\lim_{n\to\infty}U^R(\hat{\Pi}^n) \leq V_{m-1}, \] contradicting \eqref{eqn:monotone-payoff}.

The proof that \[ \liminf_{n\to\infty} \bigl(1-\chi_{m-1}^\theta(\Pi^n)\bigr)>0 \] is symmetric. If $\chi_{m-1}^\theta(\Pi^n)\to1$ for one state of nature, the two-sided odds-ratio bound implies the same convergence under the other state. Delete state $m$, make state $m-1$ the high-action absorbing state, and redirect every initial or transition probability assigned to state $m$ to state $m-1$. The resulting $(m-1)$-state protocol has a payoff differing from that of $\Pi^n$ by at most \[ p\bigl(1-\chi_{m-1}^H(\Pi^n)\bigr) + (1-p)\bigl(1-\chi_{m-1}^L(\Pi^n)\bigr), \] which converges to zero, again contradicting $\sup_{\Pi}U^R(\Pi)>V_{m-1}$. \end{proof}

For a transient state $i\in \{2,...,m-1\}$, let \[ \iota_i^\theta(\Pi^n) := {\Pr}^\theta \left( \text{state $i$ is reached} \right). \] For $x\in\mathbb{R}$, write \[ [x]_+:=\max\{x,0\}. \] We first record a consequence of asymptotic optimality. Let \[ g_{\mathrm{abs}}^n:=g^n(1)+g^n(m) \] denote the probability that $\Pi^n$ starts in an absorbing state. Since a protocol starting in state $1$ or $m$ yields at most $\bar p=\max\{p,1-p\}$, while a protocol starting in any state yields at most $\sup_{\Pi}U^R(\Pi)$, we have \[ U^R(\Pi^n) \leq g_{\mathrm{abs}}^n\bar p + \bigl(1-g_{\mathrm{abs}}^n\bigr)\sup_{\Pi}U^R(\Pi). \] Because $\sup_{\Pi}U^R(\Pi)>\bar p$ and $U^R(\Pi^n)\to \sup_{\Pi}U^R(\Pi)$, it follows that \begin{equation}\label{eqn:initial-absorption-vanishes} g_{\mathrm{abs}}^n\to0. \end{equation} 

\begin{lemma}\label{lem:on-path-continuation-order} Suppose the assumptions of \Cref{thm:polarization} hold. Then, for every $\theta\in\{H,L\}$ and every $i\in\{2,\ldots,m-1\}$, \[ \lim_{n\to\infty} \iota_i^\theta(\Pi^n) \left[ \chi_2^\theta(\Pi^n)-\chi_i^\theta(\Pi^n) \right]_+ =0 \] and \[ \lim_{n\to\infty} \iota_i^\theta(\Pi^n) \left[ \chi_i^\theta(\Pi^n)-\chi_{m-1}^\theta(\Pi^n) \right]_+ =0. \] 
\end{lemma}

\begin{proof} We prove the first claim; the second follows symmetrically. The case $i=2$ is immediate. Fix $i>2$. Starting from state $i$, let $a$, $b$, and $c$ denote the probabilities that the protocol first reaches $m$, $2$, and $1$, respectively. Then \[ a+b+c=1 \] and, by the strong Markov property, 
\[ \chi_i^{\theta}(\Pi^n)=a+b\chi_2^{\theta}(\Pi^n). \] 
If $\chi_i^{\theta}(\Pi^n)<\chi_2^{\theta}(\Pi^n)$, then 

\[ b\leq\frac{\chi_i^{\theta}(\Pi^n)}{\chi_2^{\theta}(\Pi^n)} \] and therefore \[ \begin{aligned} c &= 1-\chi_i^{\theta}(\Pi^n)-b(1-\chi_2^{\theta}(\Pi^n))\\ &\geq 1-\chi_i^{\theta}(\Pi^n) - \frac{\chi_i^{\theta}(\Pi^n)}{\chi_2^{\theta}(\Pi^n)}(1-\chi_2^{\theta}(\Pi^n))\\ &= 1-\frac{\chi_i^{\theta}(\Pi^n)}{\chi_2^{\theta}(\Pi^n)}. \end{aligned} \] Thus, conditional on reaching state $i$, the probability of being absorbed in state $1$ before returning to state $2$ is at least \[ \left[ 1-\frac{\chi_i^{\theta}(\Pi^n)}{\chi_2^{\theta}(\Pi^n)} \right]_+. \] 

Let $B_{1}^{\theta,n}$ denote the event that absorption occurs in state $1$ without state $2$ being the memory state immediately before absorption. By part \textup{(i)} of \Cref{thm:behavioral_implication} and \eqref{eqn:initial-absorption-vanishes}, \[ {\Pr}^{\theta}(B_{1}^{\theta,n})\to0. \] The preceding argument implies \[ \iota_i^\theta(\Pi^n) \left[ 1- \frac{\chi_i^\theta(\Pi^n)} {\chi_2^\theta(\Pi^n)} \right]_+ \leq {\Pr}^{\theta}(B_{1}^{\theta,n}). \] Multiplying by $\chi_2^\theta(\Pi^n)$ and using \Cref{lem:continuation-nondegeneracy} yields \[ \iota_i^\theta(\Pi^n) \left[ \chi_2^\theta(\Pi^n)-\chi_i^\theta(\Pi^n) \right]_+ \to0. \] For the second claim, use the analogous event that absorption occurs in state $m$ without state $m-1$ being the memory state immediately before absorption, together with the fact that $1-\chi_{m-1}^\theta(\Pi^n)$ is bounded away from zero. \end{proof}

\begin{proof}[\textbf{Proof of \Cref{thm:polarization}}] 

Fix $\theta\in\{H,L\}$. For each $n$, consider two independent copies of $\Pi^n$ driven by the same infinite signal sequence. Let $X_{r,t}^n$ denote the memory state of copy $r\in\{1,2\}$ at time $t$, and let \[ T_r^n \] and \[ K_r^n\in\{1,m\} \] denote its absorption time and absorbing state, respectively. Define \[ D^n:=\{K_1^n\neq K_2^n\} \] and \[ A^n:=\{T_1^n<T_2^n\}. \] By definition, \[ \lambda^\theta(\Pi^n)={\Pr}^\theta(D^n), \] where the probability integrates over both the common signal sequence and the two copies' independent internal randomizations. Let \[ \varepsilon^n := \max_{\substack{2\leq i\leq m-1\\s\in S}} f^n(i,s)(\{1,m\}). \] By part \textup{(ii)} of \Cref{thm:behavioral_implication}, \[ \varepsilon^n\to0. \] Conditional on the common signal sequence and the first copy's path through its absorption time, the probability that the second copy is absorbed at exactly the same positive time is at most $\varepsilon^n$. Together with \eqref{eqn:initial-absorption-vanishes}, this gives \[ {\Pr}^\theta(T_1^n=T_2^n) \leq \bigl(g_{\mathrm{abs}}^n\bigr)^2+\varepsilon^n \to 0. \] By exchangeability of the two copies, \[ {\Pr}^\theta(A^n) = \frac{ 1-{\Pr}^\theta(T_1^n=T_2^n) }{2} \to\frac12. \] For $j\in\{1,m\}$, define 
\[ q_j^{\theta,n} := {\Pr}^\theta(K_1^n=j\mid A^n). \] Conditional on $A^n$ and $K_1^n=j$, let \[ \psi_{i,j}^{\theta,n} := {\Pr}^\theta \left( X_{2,T_1^n}^n=i \mid A^n,\ K_1^n=j \right), \qquad i\in\{2,\ldots,m-1\}. \] Since the second copy is still transient when the first copy is absorbed, \[ \sum_{i=2}^{m-1}\psi_{i,j}^{\theta,n}=1. \] 

For this calculation, let $\mathcal F_t$ contain the
common signals $s_0,\ldots,s_{t-1}$, both initial draws,
and both copies' transition randomizations used to
generate their memory states through time $t$.
It does not contain future signals or unused randomizations.

On $A^n$, the first copy's absorbing state and the second
copy's current transient state are known at $T_1^n$.
Conditional on $\theta$, future signals and randomizations
are independent of this stopped history. Thus, on $A^n$,
\[
\Pr^\theta(K_2^n=m\mid\mathcal F_{T_1^n})
=
\chi_{X_{2,T_1^n}^n}^\theta(\Pi^n).
\]
Conditioning further on $K_1^n$ and the second copy's
current state gives the following identity.

Conditional on $A^n$, polarization occurs if the second copy is eventually absorbed in state $m$ after the first copy is absorbed in state $1$, or conversely. Therefore, \begin{equation}\label{eqn:conditional-polarization} \begin{aligned} {\Pr}^\theta(D^n\mid A^n) ={}& q_{1}^{\theta,n} \sum_{i=2}^{m-1} \psi_{i,1}^{\theta,n}\chi_i^\theta(\Pi^n)\\ &+ q_{m}^{\theta,n} \sum_{i=2}^{m-1} \psi_{i,m}^{\theta,n} \bigl(1-\chi_i^\theta(\Pi^n)\bigr). \end{aligned} \end{equation} Rewrite the right-hand side as \[ \begin{aligned} {\Pr}^\theta(D^n\mid A^n) ={}& q_{1}^{\theta,n}\chi_2^\theta(\Pi^n) + q_{m}^{\theta,n} \bigl(1-\chi_{m-1}^\theta(\Pi^n)\bigr)\\ &+ q_{1}^{\theta,n} \sum_{i=2}^{m-1} \psi_{i,1}^{\theta,n} \bigl( \chi_i^\theta(\Pi^n)-\chi_2^\theta(\Pi^n) \bigr)\\ &+ q_{m}^{\theta,n} \sum_{i=2}^{m-1} \psi_{i,m}^{\theta,n} \bigl( \chi_{m-1}^\theta(\Pi^n)-\chi_i^\theta(\Pi^n) \bigr). \end{aligned} \] For each $i$, \[ \begin{aligned} q_{1}^{\theta,n}\psi_{i,1}^{\theta,n} &= {\Pr}^\theta \left( K_1^n=1,\, X_{2,T_1^n}^n=i \mid A^n \right)\\ &\leq \frac{ \iota_i^\theta(\Pi^n) }{ {\Pr}^\theta(A^n) }. \end{aligned} \] Hence, \[ \begin{aligned} &q_{1}^{\theta,n}\psi_{i,1}^{\theta,n} \bigl( \chi_i^\theta(\Pi^n)-\chi_2^\theta(\Pi^n) \bigr)\\ &\qquad\geq - \frac{ \iota_i^\theta(\Pi^n) \left[ \chi_2^\theta(\Pi^n)-\chi_i^\theta(\Pi^n) \right]_+ }{ {\Pr}^\theta(A^n) }. \end{aligned} \] By \Cref{lem:on-path-continuation-order} and ${\Pr}^\theta(A^n)\to\tfrac12$, the negative part of this term converges to zero. 

The same argument gives \[ \begin{aligned} &q_{m}^{\theta,n}\psi_{i,m}^{\theta,n} \bigl( \chi_{m-1}^\theta(\Pi^n)-\chi_i^\theta(\Pi^n) \bigr)\\ &\qquad\geq - \frac{ \iota_i^\theta(\Pi^n) \left[ \chi_i^\theta(\Pi^n)-\chi_{m-1}^\theta(\Pi^n) \right]_+ }{ {\Pr}^\theta(A^n) }, \end{aligned} \] whose negative part also converges to zero. Because there are only finitely many transient states, it follows from \eqref{eqn:conditional-polarization} that \[ \begin{aligned} \liminf_{n\to\infty} {\Pr}^\theta(D^n\mid A^n) \geq \liminf_{n\to\infty} \min\left\{ \chi_2^\theta(\Pi^n), 1-\chi_{m-1}^\theta(\Pi^n) \right\}. \end{aligned} \] By \Cref{lem:continuation-nondegeneracy}, the right-hand side is strictly positive. Denote it by $\eta^\theta>0$. Finally, exchangeability implies \[ {\Pr}^\theta(D^n,T_1^n<T_2^n) = {\Pr}^\theta(D^n,T_2^n<T_1^n). \] Therefore, \[ \begin{aligned} \lambda^\theta(\Pi^n) = {\Pr}^\theta(D^n) &\geq 2{\Pr}^\theta(A^n) {\Pr}^\theta(D^n\mid A^n). \end{aligned} \] Since \[ 2{\Pr}^\theta(A^n)\to1, \] we conclude that \[ \liminf_{n\to\infty}\lambda^\theta(\Pi^n) \geq \eta^\theta >0. \] 
\end{proof}

\section{Non-Simple Protocols}\label{sec:non-simple}

We show that a protocol that is not effectively simple yields at
most the receiver's optimal value with one fewer memory state.
Consequently, whenever learning is valuable with $m$ states, such
protocols are bounded away from the receiver's optimal value, and 
all asymptotically optimal protocols are effectively simple. 

Recall that, for a general protocol $\Pi$, $\overline U^R(\Pi):=
\sup_{\sigma\in\br(\Pi)}U^R(\Pi,\sigma)$, and $V_k$ denote the receiver's optimal value with at most $k$
memory states, with $V_1=\bar p$.

\begin{definition}
A protocol $\Pi=(f,g,a)$ is \textbf{effectively simple} if it is
simple or satisfies both of the following conditions:
\begin{enumerate}
    \item It has exactly one absorbing state, and all other
    states are transient.
    \item Some transient state $i$ satisfies
    $a(i)\geq a(j)$ for every $j\in M$.
\end{enumerate}
\end{definition}

The terminology reflects a payoff equivalence. At a transient
state $i$ with a globally maximal action probability, stopping
is optimal for the sender under either state of nature.
Making $i$ absorbing therefore leaves the sender's optimal payoff
and the receiver's best payoff among sender best responses
unchanged. The resulting protocol is simple, as in the proof
of \Cref{lem:1-to-2-absorbing}.

The main result in this section is as follows.
\begin{theorem}\label{thm:near-optimal-effectively-simple}
Suppose $\gamma^{m-2}>\kappa$. If a sequence
$\{\Pi^n\}_{n=1}^\infty$ of protocols with $m$ memory
states satisfies $\overline U^R(\Pi^n)\to V_m$, then $\Pi^n$
is effectively simple for all sufficiently large $n$.
\end{theorem}

The proof is based on the following lemma, which is useful in its own right. 

\begin{lemma}\label{lem:not-effectively-simple}
Let $\Pi$ have $m\geq2$ memory states and suppose that it is not
effectively simple. Then
\[
\overline U^R(\Pi)
\leq V_{m-1}
=
\begin{cases}
\displaystyle
\bar p
+
\underline p\,
\frac{
\left(\gamma^{(m-3)/2}-\sqrt{\kappa}\right)^2
}{
\gamma^{m-3}-1
},
&
\textup{if }\gamma^{m-3}>\kappa,
\\[1.3em]
\bar p,
&
\textup{otherwise.}
\end{cases}
\]
In particular, the bound holds for $U^R(\Pi,\sigma)$ under every
$\sigma\in\br(\Pi)$. If $\gamma^{m-2}>\kappa$, then
\[
V_m-\overline U^R(\Pi)
\geq V_m-V_{m-1}>0.
\]
\end{lemma}

\begin{proof}
By \Cref{lem:receiver-optimal-stops},
$\overline U^R(\Pi)$ is attained by a sender best response under
which the game terminates almost surely. It therefore suffices
to bound the receiver's payoff under such a response.

As in the main text, all communicating classes below refer to the Markov chain
induced by continued information provision. Full support of the
signal distributions implies that these classes are the same
under both states of nature. Every finite chain has at least
one recurrent communicating class. Since $\Pi$ is not effectively
simple, the following cases are exhaustive.

\begin{enumerate}

\item \textbf{A recurrent class contains more than one state.}

Let $C$ be such a class and define
\[
a_C:=\max_{j\in C}a(j).
\]
Once $C$ is reached, the sender can attain $a_C$ by continuing
until a state achieving this maximum is reached and then stopping.
Because $C$ is finite, closed, and irreducible, this occurs in
finite time almost surely. No strategy can yield a higher payoff
after entry into $C$.

We may therefore select a terminating best response
$\widehat\sigma$ that follows this rule within $C$ and attains
$\overline U^R(\Pi)$. Conditional on entry into $C$, the terminal
probability of action $H$ is then $a_C$ under either state of
nature. Thus, the states in $C$ can be replaced by a single
absorbing state without changing this terminal action lottery.

Introduce a new state $c$ and let
$M':=(M\setminus C)\cup\{c\}$. Define $\Pi'=(f',g',a')$ by
\[
f'(i,s)(j)
=
\begin{cases}
f(i,s)(j),
& i,j\in M\setminus C,
\\[0.3em]
\displaystyle\sum_{k\in C}f(i,s)(k),
& i\in M\setminus C,\ j=c,
\\[0.8em]
1,
& i=j=c,
\\
0,
& i=c,\ j\in M\setminus C,
\end{cases}
\]
and
\[
g'(j)=g(j),\qquad a'(j)=a(j)
\quad\text{for }j\in M\setminus C,
\]
with
\[
g'(c)=\sum_{j\in C}g(j),
\qquad
a'(c)=a_C.
\]
Thus, every transition into $C$ is redirected to the new absorbing
state, while all other transitions and action probabilities are
unchanged.

Let $\widehat\sigma'$ agree with $\widehat\sigma$ outside $C$.
The resulting profiles induce the same terminal action
probabilities conditional on either state of nature. Moreover,
$\widehat\sigma'$ remains a best response: any sender deviation
in $\Pi'$ can be implemented in $\Pi$ by making the same
stopping decisions outside $C$ and waiting for a maximal-action
state after entering $C$. This gives the sender the same payoff.
A profitable deviation in $\Pi'$ would therefore contradict
the optimality of $\widehat\sigma$ in $\Pi$.

Since both selected profiles terminate almost surely,
\[
\overline U^R(\Pi)
=
U^R(\Pi',\widehat\sigma')
\leq\overline U^R(\Pi').
\]
The reduced protocol has
\[
|M'|=m-|C|+1\leq m-1,
\]
so $\overline U^R(\Pi)\leq V_{m-1}$.

\item \textbf{Every recurrent class is a singleton, and there is
exactly one absorbing state.}

Let $b$ be the absorbing state. All other states are transient.
Since $\Pi$ is not effectively simple, no transient state attains
the global maximum of the action rule. Hence
\[
a(j)<a(b)\qquad\text{for every }j\neq b.
\]
Continuing until absorption in $b$ gives the sender payoff $a(b)$
under either state of nature. This is the highest action
probability in the protocol, so it is optimal. The game terminates
almost surely, and every sender best response induces action $H$
with probability $a(b)$ conditional on either state. Therefore,
\[
\overline U^R(\Pi)
=
p\,a(b)+(1-p)(1-a(b))
\leq\bar p
\leq V_{m-1}.
\]

\item \textbf{Every recurrent class is a singleton, and there are
at least three absorbing states.}

Let $r\geq3$ be the number of absorbing states. By
\Cref{lem:reduce-absorbing-memory}, all but two of these states
can be eliminated while preserving the receiver's payoff under
sender best responses. The construction retains absorbing states
with the lowest and highest action probabilities and replaces
each eliminated state by a lottery over the retained states
that implements the same terminal action probability.

The resulting protocol $\Pi'$ has at most
$m-(r-2)\leq m-1$ states. Hence
\[
\overline U^R(\Pi)
\leq\overline U^R(\Pi')
\leq V_{m-1}.
\]

\end{enumerate}

The remaining possibility, exactly two singleton recurrent
classes and no other recurrent classes, would make $\Pi$ simple,
contrary to the hypothesis. Thus, the bound holds in every case.

The displayed formula for $V_{m-1}$ follows from
\Cref{thm:supUR} when $m\geq3$; when $m=2$, the reduced benchmark
is $V_1=\bar p$. Finally, if $\gamma^{m-2}>\kappa$, the same
value characterization implies $V_m>V_{m-1}$, establishing the
strict payoff gap.
\end{proof}

\begin{proof}[\textbf{Proof of \Cref{thm:near-optimal-effectively-simple}}]
Suppose, to the contrary, the claim is not true. Then along an infinite subsequence,
\Cref{lem:not-effectively-simple} gives
$\overline U^R(\Pi^n)\leq V_{m-1}<V_m$,
contradicting asymptotic optimality.
\end{proof}

\end{document}